\documentclass[12pt]{article}%
\usepackage{amssymb}
\usepackage{amsfonts}
\usepackage{graphicx, epsfig}
\usepackage{amsmath}
\usepackage{amsthm}
\usepackage{morefloats}
\usepackage{placeins}
\usepackage{rotating}
\usepackage{natbib}
\usepackage{url} 

\newcommand{\blind}{1}

\newtheorem{theorem}{Theorem}

\newtheorem{condition}{Condition}

\newtheorem{lemma}{Lemma}

\newtheorem{proposition}{Proposition}

\DeclareMathOperator{\Tr}{Tr}

\begin{document}

\def\spacingset#1{\renewcommand{\baselinestretch}%
{#1}\small\normalsize} \spacingset{1}


\if1\blind
{
  \title{\bf Forecasting Multiple Time Series with One-Sided Dynamic Principal Components}
  \author{Daniel Pe\~{n}a\thanks{
\textit{Daniel Pe\~{n}a is Professor, Department of Statistics and Institute of Financial Big Data, Universidad Carlos III de Madrid,
Calle Madrid 126, 28903 Getafe, Espa\~{n}a, (e-mail: daniel.pena@uc3m.es). Ezequiel Smucler is currently Postdoctoral Research Fellow, Department of Statistics, University of British Columbia, 3182 Earth Sciences Building, 2207 Main Mall
Vancouver, BC, Canada V6T 1Z4 (e-mail: esmucler@cs.ubc.ca). Victor J.
Yohai is Professor Emeritus, Mathematics Department, Faculty of Exact Sciences,
Ciudad Universitaria, 1428 Buenos Aires, Argentina (e-mail: victoryohai@gmail.com). D.P. has been supported by Grant ECO2015-66593-P of MINECO/FEDER/UE. E.S. was partially funded by a CONICET Ph.D fellowship and by grant PIP 112-201101-00339 from CONICET.}}\hspace{.2cm}\\
    Department of Statistics and Institute of Financial Big Data\\
    Universidad Carlos III de Madrid, Spain\\
    and \\
    Ezequiel Smucler \\
    Instituto de Calculo\\
    School of Exact and Natural Sciences\\
    Universidad de Buenos Aires - CONICET, Argentina\\
	and \\
    Victor J. Yohai \\
    Instituto de Calculo and Department of Mathematics\\
    School of Exact and Natural Sciences\\
    Universidad de Buenos Aires - CONICET, Argentina
    }
    \date{}
  \maketitle
} \fi

\if0\blind
{
  \bigskip
  \bigskip
  \bigskip
  \begin{center}
    {\LARGE\bf Forecasting Multiple Time Series with One-Sided Dynamic Principal Components}
\end{center}
    \date{}
  \medskip
} \fi

\bigskip
\begin{abstract}
We define one-sided dynamic principal components (ODPC) for time series as
linear combinations of the present and past values of the series that minimize
the reconstruction mean squared error. Previous definitions of dynamic
principal components depend on past and future values of the series. For this
reason, they are not appropriate for forecasting purposes. On the contrary, it
is shown that the ODPC introduced in this paper can be successfully used for
forecasting high-dimensional multiple time series. An alternating least
squares algorithm to compute the proposed ODPC is presented. We prove that for
stationary and ergodic time series the estimated values converge to their
population analogues. We also prove that asymptotically, when both the number
of series and the sample size go to infinity, if the data follows a dynamic
factor model, the reconstruction obtained with ODPC\ converges, in mean
squared error, to the common part of the factor model. Monte Carlo results
shows that forecasts obtained by the ODPC compare favourably with other
forecasting methods based on dynamic factor models.
\end{abstract}

\noindent%
\vfill

\newpage
\spacingset{1.45} 
\section{Introduction}
\label{sec:intro}

Forecasting a large number of cross-correlated time series is a difficult
problem. Building a multivariate VARMA model is only possible when the number
of series is small compared to the sample size. Therefore, other alternatives
have been explored. \cite{BoxTiao} introduced linear combinations of the
series with maximum predictability. \cite{Litterman} proposed Bayesian VAR
models with shrinking prior distributions to control the number of parameters.
\cite{AhnReinsel} addressed this problem by introducing reduced-rank
autoregressive models. \cite{TiaoTsay} presented ways to simplify the
construction of VARMA models by identifying scalar components. However, the
currently most popular procedures for large data sets are based on dynamic
factor models, where the relationship between the series and the factor can be
contemporaneous, or with lags. \cite{SW2002} use the contemporaneous model for
forecasting assuming that all the variables follow the same dynamic factor
model. Then, the forecast of a given variable can be written as the sum of the
forecast of the common component, driven by the factors, plus the univariate
forecasts of the idiosyncratic component. They used principal components of
the explanatory variables to obtain consistent estimators of the factor
effects and fitted univariate autoregressive models to forecast the
idiosyncratic component. Their method showed a good performance in simulated
and real macroeconomic data. This work also explains why univariate forecasts
are improved by using as a regressor a weighted average of all the series.
This procedure, called forecast pooling (see \cite{GarciaFerrer}), can be
justified by assuming a common factor in the series, as shown by
\cite{Pena2004}.

\cite{Forni2000} proposed a general dynamic factor model assuming lagged
relationships between the series and the factors. They allow for an infinite
number of factor lags and low correlation between any two idiosyncratic
components. \cite{Forni2005} proposed a one-sided method of estimation of the
common part of a dynamic factor model for forecasting. The forecasts generated
with this procedure have been compared tothe ones derived by \cite{SW2002} and
the results are mixed (see \citet{Forni2016}). A modified forecasting approach
was proposed by \cite{Forni2015}, although again, as shown in \cite{Forni2015}%
, the results are mixed.

\cite{PenaYohai2016}, following Brillinger's idea of dynamic principal
components, \cite{Brillinger64, Brillinger}, proposed components that provide
an optimal reconstruction of the series in finite samples, but dropping
Brillingers's assumption that the components are linear combinations of the
data. However, this approach is not expected to work well in forecasting
problems as the last values of the dynamic principal components have been
computed with smaller number of observations than the central values.

In Section \ref{sec:Comp} of this paper we define one-sided dynamic principal
components (ODPC) as linear combinations of present and previous values of the
series which have optimal reconstruction performance, that is, they minimize a
mean squared error reconstruction criterion. Following Hotelling's original
spirit of principal components, our definition is not based on assuming any
model for the vector time series. We show how to forecast future values of the
series using the proposed ODPC and a univariate forecasting method. We present
two properties of the proposed estimator. In Section \ref{sec:cons} we prove
that for stationary and ergodic time series the estimated values converge to
their population analogues. We also prove, in Section \ref{sec:cons:dfm}, that
asymptotically, when both the number of series and the sample size goes to
infinity, if the data follows a dynamic factor model, the reconstruction
obtained with ODPC\ converges, in mean squared error, to the common
part of the factor model. In Sections \ref{sec:simu} and \ref{sec:real_data}
we illustrate with Monte Carlo simulations and with a real data example that
our forecasting procedure compares favourably with other forecasting methods
based on dynamic factor models. We discuss possible strategies for
choosing the number of components and lags used to define them in Section \ref{sec:lags}.
Finally, some conclusions and possible
extensions are discussed in Section \ref{sec:conclu}. Section \ref{sec:appen}
is a technical appendix containing the proofs of our main theoretical results.
\section{One Sided Dynamic Components and their computation}

\label{sec:Comp}

Consider the vector time series $\mathbf{z}_{1},\dots,\mathbf{z}_{T},$ where
$\mathbf{z}_{t}=(z_{t,1},\dots,z_{t,m})^{\prime}$. Let $\mathbf{Z}$ be the
data matrix of dimension $T\times m$ where each row is $\mathbf{z}_{t}%
^{\prime}$. Consider integer numbers $k^{1}_{1},k^{1}_{2}\geq0$. Let $\mathbf{a}%
=(\mathbf{a}_{0}^{\prime},\dots,\mathbf{a}_{k^{1}_{1}}^{\prime})^{\prime}$, where
$\mathbf{a}_{h}^{\prime}=(a_{h,1},...,a_{h,m})$, be a vector of dimension
$m(k_{1}+1)\times1$, let $\boldsymbol{\alpha}^{\prime}=(\alpha_{1}%
,\dots,\alpha_{m})$ and $\mathbf{B}$ the matrix that has coefficients
$b_{h,j}$ and dimension $(k^{1}_{2}+1)\times m$. We can define the first one-sided
dynamic principal component with $k_{1}^{1}$ lags as the vector
\begin{equation}
f_{t}=\sum\limits_{j=1}^{m}\sum\limits_{h=0}^{k_{1}^{1}}a_{h,j}z_{t-h,j}\quad
t=k_{1}^{1}+1,\dots,T, \label{eq:component}%
\end{equation}
and use this component to reconstruct the series using $k_{2}^{1}$ lags of the
component as
\begin{equation}
z_{t,j}^{R}(\mathbf{a},\boldsymbol{\alpha},\mathbf{B)}=\alpha_{j}%
+\sum\limits_{h=0}^{k_{2}^{1}}b_{h,j}f_{t-h}.\nonumber
\end{equation}
The values $(k_{1}^{1},k_{2}^{1})$ that define the first dynamic principal
components will be discussed later. Suppose now they are given. Then, the
theoretical optimal values of $\mathbf{a},\boldsymbol{\alpha}$ and
$\mathbf{B}$ can be defined as those that minimize the mean squared error in
the reconstruction of the data, that is, calling $\mathbb{E}$ the expectation
operator, as the solutions of
\begin{equation}
(\mathbf{a}^{\ast},\boldsymbol{\alpha}^{\ast},\mathbf{B}^{\ast}\mathbf{)=}%
\arg\min_{\mathbf{a},\boldsymbol{\alpha}\text{,}\mathbf{B}}\frac{1}%
{T-(k_{1}^{1}+k_{2}^{2})}\sum\limits_{j=1}^{m}\sum\limits_{t=(k_{1}^{1}%
+k_{2}^{2})+1}^{T}\mathbb{E}\left(  z_{t,j}-z_{t,j}^{R}(\mathbf{a}%
,\boldsymbol{\alpha},\mathbf{B)}\right)  ^{2}\nonumber
\end{equation}

Natural estimators of $(\mathbf{a}^{\ast},\boldsymbol{\alpha}^{\ast
},\mathbf{B}^{\ast}\mathbf{)}$ can be defined as solutions of
\begin{equation}
\arg\min_{\mathbf{a},\boldsymbol{\alpha},\mathbf{B}}\text{MSE}(\mathbf{a}%
,\boldsymbol{\alpha},\mathbf{B}) \label{eq:def_estim_partial}%
\end{equation}
where
\begin{equation}
\text{MSE}(\mathbf{a},\boldsymbol{\alpha},\mathbf{B})=\frac{1}{T-(k_{1}%
^{1}+k_{2}^{1})}\sum\limits_{j=1}^{m}\sum\limits_{t=(k_{1}^{1}+k_{2}^{1}%
)+1}^{T}(z_{t,j}-z_{t,j}^{R}(\mathbf{a},\boldsymbol{\alpha},\mathbf{B)})^{2}
\label{eq:MSE}%
\end{equation}
Note that if $(\mathbf{a},\boldsymbol{\alpha},\mathbf{B)}$ is a solution of
\eqref{eq:def_estim_partial} then $(\gamma\mathbf{a},\boldsymbol{\alpha
},\mathbf{B/\gamma)}$ will be one as well. Hence, we define the optimal
$\widehat{\mathbf{a}}$, $\widehat{\boldsymbol{\alpha}}$ and
$\widehat{\mathbf{B}}$ as any solution of
\begin{equation}
\text{MSE}(\widehat{\mathbf{a}},\widehat{\boldsymbol{\alpha}}%
,\widehat{\mathbf{B}})=\min_{\Vert\mathbf{a}\Vert=1,\boldsymbol{\alpha
},\mathbf{B}}\text{MSE}(\mathbf{a},\boldsymbol{\alpha},\mathbf{B}).
\label{eq:def_estim}%
\end{equation}
Conditions to guarantee the existence of at least one solution of
\eqref{eq:def_estim} will be given in Section \ref{sec:cons}.

Let
\begin{equation}
\widehat{f}_{t}=\sum\limits_{j=1}^{m}\sum\limits_{h=0}^{k_{1}^{1}}%
\widehat{a}_{h,j}z_{t-h,j}, \label{eqforf}%
\end{equation}
and $\widehat{z}_{t,j}=z_{t,j}^{R}(\widehat{\mathbf{a}}%
,\widehat{\boldsymbol{\alpha}},\widehat{\mathbf{B}})=\widehat{\alpha}_{j}%
+\sum_{h=0}^{k_{2}^{1}}\widehat{b}_{h,j}\widehat{f}_{t-h}$ be the
corresponding optimal reconstruction of the $j$-th series at period $t$,
$1\leq j\leq m$, $(k_{1}^{1}+k_{2}^{1})+1\leq t\leq T$. We define the second
one-sided dynamic principal component with $(k_{1}^{2},k_{2}^{2})$ lags as the
first one-sided dynamic principal component of the residuals $z_{t,j}%
-\widehat{z}_{t,j}$, $1\leq j\leq m$, $(k_{1}^{1}+k_{2}^{1})+1\leq t\leq T$.
Higher order principal components are defined similarly. Note that if we
compute $q$ one-sided principal components, each with $(k_{1}^{i},k_{2}^{i})$
lags, $1\leq i\leq q$, we will only be able to reconstruct the periods
$\sum_{i=1}^{q}\left(  k_{1}^{i}+k_{2}^{i}\right)  +1,\dots,T$. The
superscript $i$, indicating the principal component in the vector of lags
$(k_{1}^{i},k_{2}^{i})$ will generally be omitted when no confusion could arise.

In order to derive an algorithm to compute estimators of $\mathbf{a}%
,\boldsymbol{\alpha}$ and $\mathbf{B}$, we need first to express the objective
function, $\text{MSE}(\mathbf{a},\boldsymbol{\alpha},\mathbf{B})$, in a more
manageable form. To this end, we will introduce further notation. Throughout
this paper $\Vert\cdot\Vert$ will stand for the Euclidean norm for vectors and
the spectral norm for matrices, whereas $\Vert\cdot\Vert_{F}$ will stand for
the Frobenius norm for matrices. $\mathbf{A}^{\dagger}$ will stand for the
Moore-Penrose pseudo-inverse of a matrix $\mathbf{A}$. For $h\ =0,\dots
,(k_{1}+k_{2})$, let $\mathbf{Z}_{h\ }$ be the $(T-(k_{1}+k_{2}))\times m$
data matrix
\begin{equation}
\mathbf{Z}_{h}=%
\begin{pmatrix}
\mathbf{z}_{h+1}^{\prime}\\
\mathbf{z}_{h+2}^{\prime}\\
\mathbf{z}_{T-(k_{1}+k_{2})+h}^{\prime}%
\end{pmatrix}
\end{equation}
For $l=k_{1},\dots,(k_{1}+k_{2})$, let $\mathbf{Z}_{l,0}=\left[
\mathbf{Z}_{l},\mathbf{Z}_{l-1},...,\mathbf{Z}_{l-k_{1}}\right]  $ be a
$(T-(k_{1}+k_{2}))\times m(k_{1}+1)$ matrix. Then, letting $\mathbf{g}%
_{l}=(f_{l+1},...,f_{T-(k_{1}+k_{2})+l})^{\prime}$ be a vector of dimension
$(T-(k_{1}+k_{2}))\times1$ we have $\mathbf{g}_{l}=\mathbf{Z}_{l,0}\mathbf{a}%
$. Let $\mathbf{D}$ be the matrix of dimension $(k_{2}+2)\times m$ given by
$\mathbf{D}=%
\begin{pmatrix}
\boldsymbol{\alpha}^{\prime}\\
\mathbf{B}%
\end{pmatrix}
$.

The reconstruction of the values of the $\mathbf{Z}_{k_{1}+k_{2}}$ matrix
using $\mathbf{a}$, $\boldsymbol{\alpha}$ and $\mathbf{B}$ can be written as a
matrix $\widehat{\mathbf{Z}}_{k_{1}+k_{2}}$ of the same dimension,
$(T-(k_{1}+k_{2}))\times m$, as
\[
\widehat{\mathbf{Z}}_{k_{1}+k_{2}}=\mathbf{F}_{k_{1},k_{2}}\mathbf{D}%
\]
where $\mathbf{F}_{k_{1},k_{2}}=\left[  \mathbf{1}_{T-(k_{1}+k_{2}%
)},\mathbf{g}_{(k_{1}+k_{2})},\mathbf{g}_{(k_{1}+k_{2})-1},...,\mathbf{g}%
_{k_{1}}\right]  $ is a matrix with dimensions $(T-(k_{1}+k_{2}))\times
(k_{2}+2)$ and $\mathbf{1}_{T-(k_{1}+k_{2})}$ is a vector of length
$T-(k_{1}+k_{2})$ with all its coordinates equal to one. Note that
$\mathbf{F}_{k_{1},k_{2}}=\mathbf{F}_{k_{1},k_{2}}(\mathbf{a})$, even though
this dependence will not in general be made explicit in the notation.

Let $\mathbf{C}$ be the matrix with dimensions $(T-(k_{1}+k_{2}))(k_{2}%
+1)\times m(k_{1}+1)$ given by
\begin{equation}
\mathbf{C}=\left(
\begin{array}
[c]{c}%
\mathbf{Z}_{k_{1}+k_{2},0}\\
\vdots\\
\mathbf{Z}_{k_{1},0}%
\end{array}
\right)  . \label{eq:def:C}%
\end{equation}
Note that
\[
\mathbf{F}_{k_{1},k_{2}}=\left(  \mathbf{1}_{T-(k_{1}+k_{2})},\mathbf{Z}%
_{k_{1}+k_{2},0}\mathbf{a},\dots,\mathbf{Z}_{k_{1},0}\mathbf{a}\right)  .
\]
Hence
\[
vec(\mathbf{F}_{k_{1},k_{2}})^{\prime}=(\mathbf{1}_{T-(k_{1}+k_{2})}^{\prime
},(\mathbf{Z}_{k_{1}+k_{2},0}\mathbf{a})^{\prime},\dots,(\mathbf{Z}_{k_{1}%
,0}\mathbf{a})^{\prime})=%
\begin{pmatrix}
\mathbf{1}_{T-(k_{1}+k_{2})}\\
\mathbf{C}\mathbf{a}%
\end{pmatrix}
^{\prime}.
\]
and
\begin{align*}
vec(\widehat{\mathbf{Z}}_{k_{1}+k_{2}})=vec(\mathbf{F}_{k_{1},k_{2}}%
\mathbf{D})  &  =(\mathbf{D}^{\prime}\otimes\mathbf{I}_{T-(k_{1}+k_{2}%
)})vec(\mathbf{F}_{k_{1},k_{2}})\\
&  =(\mathbf{D}^{\prime}\otimes\mathbf{I}_{T-(k_{1}+k_{2})})%
\begin{pmatrix}
\mathbf{1}_{T-(k_{1}+k_{2})}\\
\mathbf{C}\mathbf{a}%
\end{pmatrix}
.
\end{align*}
Then, $(\widehat{\mathbf{a}},\widehat{\boldsymbol{\alpha}},\widehat{\mathbf{B}%
})$ can be obtained by minimizing
\[
\Vert\mathbf{Z}_{k_{1}+k_{2}}\mathbf{-}\widehat{\mathbf{Z}}_{k_{1}+k_{2}}%
\Vert_{F}^{2}=\Vert\mathbf{Z}_{k_{1}+k_{2}}-\mathbf{F}_{k_{1},k_{2}}%
\mathbf{D}\Vert_{F}^{2}=\Vert vec(\mathbf{Z}_{k_{1}+k_{2}}%
)-vec(\widehat{\mathbf{Z}}_{k_{1}+k_{2}})\Vert^{2},
\]
subject to $\Vert\widehat{\mathbf{a}}\Vert=1$. Note that
\begin{align*}
&  \Vert vec(\mathbf{Z}_{k_{1}+k_{2}})-vec(\widehat{\mathbf{Z}}_{k_{1}+k_{2}%
})\Vert^{2}=\\
&  \left\Vert vec(\mathbf{Z}_{k_{1}+k_{2}})-(\mathbf{D}^{\prime}%
\otimes\mathbf{I}_{T-(k_{1}+k_{2})})%
\begin{pmatrix}
\mathbf{1}_{T-(k_{1}+k_{2})}\\
\mathbf{C}\mathbf{a}%
\end{pmatrix}
\right\Vert ^{2}=\\
&  \left\Vert vec(\mathbf{Z}_{k_{1}+k_{2}})-(\boldsymbol{\alpha}%
\otimes\mathbf{I}_{T-(k_{1}+k_{2})})\mathbf{1}_{T-(k_{1}+k_{2})}%
-(\mathbf{B}^{\prime}\otimes\mathbf{I}_{T-(k_{1}+k_{2})})\mathbf{C}%
\mathbf{a}\right\Vert ^{2}%
\end{align*}
For a fixed $\mathbf{D}$, $\widehat{\mathbf{a}}$ can be computed by least
squares
\begin{equation}
\widehat{\mathbf{a}}=\left(  (\mathbf{B}^{\prime}\otimes\mathbf{I}%
_{T-(k_{1}+k_{2})})\mathbf{C}\right)  ^{\dagger}(vec(\mathbf{Z}_{k_{1}+k_{2}%
})-(\boldsymbol{\alpha}\otimes\mathbf{I}_{T-(k_{1}+k_{2})})\mathbf{1}%
_{T-(k_{1}+k_{2})}). \label{eqfora}%
\end{equation}
and then standardized to unit norm. On the other hand, for a fixed
$\mathbf{F}_{k_{1},k_{2}}$, the optimal $\mathbf{D}$ can also be computed by
least squares
\begin{equation}
\widehat{\mathbf{D}}=(\mathbf{F}_{k_{1},k_{2}})^{\dagger}\mathbf{Z}%
_{k_{1}+k_{2}}. \label{eqforD}%
\end{equation}
Then, $\widehat{\boldsymbol{\alpha}}$ is given by the first row of
$\widehat{\mathbf{D}}$ and $\widehat{\mathbf{B}}$ is given by the last
$k_{2}+1$ rows of $\widehat{\mathbf{D}}$.

\subsection{Computing algorithm}

We propose the following alternating Least Squares algorithm for computing
$\widehat{\mathbf{a}}$ , $\widehat{\mathbf{D}}$. Let $\mathbf{a}^{(i)}$,
$\mathbf{D}^{(i)}$ and $\mathbf{f}^{(i)}$ be the values of $\mathbf{a}$,
$\mathbf{D}$ and $\mathbf{f}$ corresponding to the $i$-th iteration. Let
$\delta\in(0,1)$, a tolerance parameter to stop the iterations. Write
$\text{MSE}(\mathbf{a},\mathbf{D})=\text{MSE}(\mathbf{a},\boldsymbol{\alpha
},\mathbf{B})$. In order to define the algorithm it is enough to give an
initial value of the component, $\mathbf{f=(}f_{k_{1}+1,}...,f_{T})^{^{\prime
}}$, say $\mathbf{f}^{(0)}$, and describe a rule to compute $\mathbf{D}%
^{(i+1)},\mathbf{a}^{(i+1)}$ and $\mathbf{f}^{(i+1)}$ from $\mathbf{f}^{(i)}$.
This can be done as follows:

\bigskip

1. Given $\mathbf{f}^{(i)}$ define $\mathbf{D}^{(i+1)}$ by (\ref{eqforD}),
where $\mathbf{F}_{k_{1},k_{2}}$ corresponds to $\mathbf{f}^{(i)}.$

2. Compute $\mathbf{a}_{\ast}^{(i+1)}$ by \eqref{eqfora} with $\mathbf{D}%
=\mathbf{D}^{(i+1)}$ and let $\mathbf{a}^{(i+1)}=\mathbf{a}_{\ast}%
^{(i+1)}/\left\Vert \mathbf{a}_{\ast}^{(i+1)}\right\Vert .$

3. The $t$-th coordinate of $\mathbf{f}^{(i+1)}$ is given by
\eqref{eq:component} with $\mathbf{a}=\mathbf{a}^{(i+1)}$.

\bigskip

The stopping rule is as follows: Stop when
\[
\frac{\text{MSE}(\mathbf{a}^{(i)},\mathbf{D}^{(i)}) - \text{MSE}%
(\mathbf{a}^{(i+1)},\mathbf{D}^{(i+1)})}{\text{MSE}(\mathbf{a}^{(i)}%
,\mathbf{D}^{(i)})}\leq\delta
\]

Clearly in this algorithm at each step the MSE\ decreases and therefore it
converges to a local minimum. To obtain a global minimum the initial value
$\mathbf{f}^{(0)}$ should be close enough to the optimal one. We propose to
take $\mathbf{f}^{(0)}$ as the last $T-k_{1}$ coordinates of the first
ordinary principal component of the data. Alternatively, the first Generalized
Dynamic Principal Component proposed by \cite{PenaYohai2016} could be used.

Note that since the matrix $(\mathbf{B}^{\prime}\otimes\mathbf{I}%
_{T-(k_{1}+k_{2})})\mathbf{C}$ has dimensions $m(T-(k_{1}+k_{2}))\times
m(k_{2}+1)$, solving the associated least squares problem can be time
consuming for high-dimensional (large $m$) problems. The iterative nature of
the algorithm we propose implies that this least squares problem will have to
be solved several times for different $\mathbf{B}$ matrices. However, note
that since the matrix $\mathbf{B}^{\prime}\otimes\mathbf{I}_{T-(k_{1}+k_{2})}$
is sparse, it can be stored efficiently, and multiplying it with a vector is
relatively fast. We found that for problems with a moderately large $m$, the
following modification of our algorithm works generally faster: instead of
finding the optimal $\mathbf{a}^{(i+1)}$ corresponding to $\mathbf{D}^{(i+1)}%
$, just do one iteration of coordinate descent for $\mathbf{a}^{(i+1)}.$

\subsection{Forecasting}

\label{subsec:forecasting}

Suppose we have fitted $q$ dynamic principal components to the data, each with
$(k_{1}^{i},k_{2}^{i})$ lags, $i=1,\dots,q$. Let $\widehat{\mathbf{f}}_{T}%
^{i}$ be the vector with the estimated values for the $i$-th dynamic principal
component and $\widehat{\mathbf{B}}^{i}$, $\widehat{\boldsymbol{\alpha}}^{i}$
be the corresponding loadings and intercepts. Suppose we have decided upon a
procedure to forecast each of these dynamic principal components separately,
and let $\widehat{f}_{T+h|T}^{i}$ for $h>0$ be the forecast of $f_{T+h}^{i}$
with information until time $T.$ We can obtain an $h$-steps ahead forecast of
$\mathbf{z}_{T}$ as
\[
\widehat{z}_{T+h|T,j}=\sum\limits_{i=1}^{q}\left(  \widehat{\alpha}_{j}%
^{i}+\sum\limits_{v=0}^{k_{2}^{i}}\widehat{b}_{v,j}^{i}\widehat{f}%
_{T+h-v|T}^{i}\right)  \quad j=1,\dots,m.
\]


\section{Wald type consistency for stationary data}

\label{sec:cons}

In this section we prove a consistency result for our procedure in the case of
stationary and ergodic vector time series. First, in Proposition
\ref{propo:exist_minsamp} we show that, asymptotically and with probability
one, problem \eqref{eq:def_estim} is well defined. In Proposition
\ref{propo:exist_minpop} we show that the population version of
\eqref{eq:def_estim}, obtained replacing means by expectations is well
defined. Finally, in Theorem \ref{theo:consistency}, we prove that the
distance between any given solution of \eqref{eq:def_estim} and the set of
solutions of the population problem converges almost surely to zero.

We will assume that the process $\mathbf{z}_{t}$, $t\geq1$, defined in a
probability space $(\Omega,\mathcal{F},\mathbb{P})$, is strictly stationary
and ergodic. We note that if instead one assumes weak second order
stationarity and ergodicity, similar results can be obtained, but with
convergences in probability instead of almost surely. Let $\mathbf{z}$ denote
a random variable with the common distribution of the $\mathbf{z}%
_{t}^{\text{'}}$s. Given a square matrix $\mathbf{M}$, let $\lambda
_{min}(\mathbf{M})$ and $\lambda_{max}(\mathbf{M})$ denote the smallest and
largest (in absolute value) eigenvalues of $\mathbf{M}$ respectively.
Moreover, $\Tr(\mathbf{M})$ will denote the trace of $\mathbf{M}$. We will
assume that $\mathbb{E}\mathbf{z}=0$ and that $\mathbb{E}\Vert\mathbf{z}%
\Vert^{2}<\infty$. Let $\boldsymbol{\Sigma}(l)=\mathbb{E}\mathbf{z}%
_{t}\mathbf{z}_{t-l}^{\prime}$ be the lag $l$ autocovariance matrix of the
process $\mathbf{z}_{t}$.

Note that \eqref{eq:MSE} can be written as
\[
\text{MSE}(\mathbf{a},\boldsymbol{\alpha},\mathbf{B})=\frac{1}{T-(k_{1}%
+k_{2})}\sum\limits_{t=(k_{1}+k_{2})+1}^{T}\Vert\mathbf{z}_{t}%
-\widehat{\mathbf{z}}_{t}\Vert^{2},
\]
where $\widehat{\mathbf{z}}_{t}$ is the vector of length $m$ with coordinates
\[
\widehat{z}_{t,j}=\widehat{\alpha}_{j}+\sum\limits_{h=0}^{k_{2}}%
\widehat{b}_{h,j}\widehat{f}_{t-h},\text{ }j=1,\dots,m.
\]
Let $\mathbf{x}_{t}^{\prime}=(\mathbf{z}_{t}^{\prime},\dots,\mathbf{z}%
_{t-k_{1}}^{\prime})$, so that given any $\mathbf{a}$ with $\Vert
\mathbf{a}\Vert=1$, the corresponding component at time $t$, $f_{t}$, is given
by $\mathbf{a}^{\prime}\mathbf{x}_{t}$. The lag $l$ autocovariance matrix of
$\mathbf{x}_{t}$ is given by
\[
\mathbf{V}(l)=%
\begin{pmatrix}
\boldsymbol{\Sigma}(l) & \boldsymbol{\Sigma}(l+1) & \dots & \boldsymbol{\Sigma
}(l+k_{1})\\
\vdots & \vdots &  & \vdots\\
\boldsymbol{\Sigma}(l-k_{1}) & \boldsymbol{\Sigma}(l-k_{1}+1) & \dots &
\boldsymbol{\Sigma}(l)
\end{pmatrix}
\]
Fix $\mathbf{a}$ with $\Vert\mathbf{a}\Vert=1$, then the covariance matrix of
the vector
\[
(1,\mathbf{a}^{\prime}\mathbf{x}_{t},\mathbf{a}^{\prime}\mathbf{x}_{t-1}%
,\dots,\mathbf{a}^{\prime}\mathbf{x}_{t-k_{2}})
\]
is
\[
\mathcal{S}(\mathbf{a})=%
\begin{pmatrix}
1 & 0 & 0 & \dots & 0\\
0 & \mathbf{a}^{\prime}\mathbf{V}(0)\mathbf{a} & \mathbf{a}^{\prime}%
\mathbf{V}(1)\mathbf{a} & \dots & \mathbf{a}^{\prime}\mathbf{V}(k_{2}%
)\mathbf{a}\\
\vdots & \vdots & \vdots &  & \vdots\\
0 & \mathbf{a}^{\prime}\mathbf{V}(-k_{2})\mathbf{a} & \mathbf{a}^{\prime
}\mathbf{V}(-k_{2}+1)\mathbf{a} & \dots & \mathbf{a}^{\prime}\mathbf{V}%
(0)\mathbf{a}%
\end{pmatrix}
.
\]
We will need the following assumption.

\begin{condition}
\label{condition:eig} There exists $\eta<1$ such that
\[
\mathbb{P}\left(  \sum\limits_{h=0}^{k_{2}}v_{h}(\mathbf{a}^{\prime}%
\mathbf{x}_{t-h})=v_{k_{2}+1}\right)  \leq\eta\
\]
for all $\mathbf{a}\in\mathbb{R}^{m(k_{1}+1)}$ with $\Vert\mathbf{a}\Vert=1$
and $\mathbf{v}=(v_{0},\dots,v_{k_{2}},v_{k_{2}+1})$ such that $\mathbf{v}%
\neq0$. That is, for any $\mathbf{a}$, there is no deterministic linear
relation between the values $f_{t},\dots,f_{t-k_{2}}$.
\end{condition}

It follows from Condition 1 that $\inf_{\Vert\mathbf{a}\Vert=1} \lambda
_{min}(\mathcal{S}(\mathbf{a}))>0$. Note that $\sup_{\Vert\mathbf{a}\Vert=1}
\lambda_{max}(\mathcal{S}(\mathbf{a}))<\infty$ always holds.

Proposition \ref{propo:exist_minsamp} shows that, asymptotically and with
probability one, there exists at least one solution of \eqref{eq:def_estim}.
Write $\text{MSE}(\mathbf{a},\mathbf{D})$ for $\text{MSE}(\mathbf{a}%
,\boldsymbol{\alpha},\mathbf{B})$, where $\mathbf{D}=%
\begin{pmatrix}
\boldsymbol{\alpha}^{\prime}\\
\mathbf{B}%
\end{pmatrix}
.$

\begin{proposition}
\label{propo:exist_minsamp} Assume Condition \ref{condition:eig} holds. Then,
with probability one, there exists $T_{0}$ such that for all $T>T_{0}$,
$\arg\min_{\Vert\mathbf{a} \Vert=1, \mathbf{D}}MSE(\mathbf{a},\mathbf{D})$ has
at least one solution.
\end{proposition}

Let
\begin{equation}
\text{MSE}_{0}(\mathbf{a},\boldsymbol{\alpha},\mathbf{B})=\mathbb{E}%
\Vert\mathbf{z}_{t}-\mathbf{z}_{t}^{R}(\mathbf{a},\boldsymbol{\alpha
},\mathbf{B})\Vert^{2}\nonumber
\end{equation}
be the population version of \eqref{eq:MSE}. It is easy to verify that
$\text{MSE}_{0}$ is continuous. Let
\[
\mathcal{I}=\left\{  (\mathbf{a}^{\ast},\mathbf{D}^{\ast}):\text{MSE}%
_{0}(\mathbf{a}^{\ast},\mathbf{D}^{\ast})=\inf_{\Vert\mathbf{a}\Vert
=1,\mathbf{D}}\text{MSE}_{0}(\mathbf{a},\mathbf{D})\right\}  .
\]
Proposition \ref{propo:exist_minpop} entails that $\mathcal{I}$ is non-empty.

\begin{proposition}
\label{propo:exist_minpop} Assume Condition \ref{condition:eig} holds. Then
$\inf_{\Vert\mathbf{a} \Vert= 1} \text{MSE}_{0}(\mathbf{a}, \mathbf{D})
\rightarrow+\infty$ when $\Vert\mathbf{D}\Vert_{F} \rightarrow+\infty$.
\end{proposition}

Let $d((\mathbf{a},\mathbf{D}),\mathcal{I})=\inf\left\{  \Vert\mathbf{a}%
-\mathbf{a}^{\ast}\Vert+\Vert\mathbf{D}-\mathbf{D}^{\ast}\Vert_{F}%
:(\mathbf{a}^{\ast},\mathbf{D}^{\ast})\in\mathcal{I}\right\}  .$

\begin{theorem}
\label{theo:consistency}[Consistency] Assume Condition \ref{condition:eig}
holds. Let $(\widehat{\mathbf{a}},\widehat{\boldsymbol{\alpha}}%
,\widehat{\mathbf{B}})$ be a solution of \eqref{eq:def_estim} and
$\widehat{\mathbf{D}}^{\prime}=%
\begin{pmatrix}
\widehat{\boldsymbol{\alpha}} & \widehat{\mathbf{B}}^{\prime}%
\end{pmatrix}
$. Then $d((\widehat{\mathbf{a}},\widehat{\mathbf{D}}),\mathcal{I}%
)\overset{\text{a.s}.}{\rightarrow}0$.
\end{theorem}

\section{Consistency in the dynamic factor model}

\label{sec:cons:dfm} In this section we deal with the interesting case where
the series follow a stationary dynamic factor model. In Theorem
\ref{theo:common} we prove a consistency result for this situation:
asymptotically, when both the number of series and the sample size go to
infinity, the reconstruction obtained with ODPC converges in mean-square to
the common part of the factor model.

Suppose we have observations, $\mathbf{z}_{1}, \dots, \mathbf{z}_{T}$,
$\mathbf{z}_{t}^{\prime}=(z_{t,1},\dots,z_{t,m})$, of a double indexed
stochastic process $\lbrace z_{t,j} : t\in\mathbb{Z}, j \in\mathbb{N} \rbrace
$. Consider the following dynamic factor model with one factor, say $f_{t},$
and a finite dimensional factor space. That is,
\[
z_{t,j}=\sum\limits_{h=0}^{k_{2}}b_{h,j}f_{t-h}+e_{t,j},\quad t=1,\dots
,T,\quad j=1,\dots,m,
\]
where the $e_{t,j}$ for $j=1,...,m$ and $f_{t}$ are stationary processes and
$b_{h,j}$ the factor loadings. This can be expressed in the form of a factor
model, with $k_{2}+1$ static factors, as
\[
\mathbf{z}_{t}=\mathbf{B}^{\prime}\mathbf{f}_{t}+\mathbf{e}_{t},\quad
t=1,\dots,T,
\]
where $\mathbf{e}_{t}=(e_{t,1},\dots,e_{t,m})^{\prime}$, $\mathbf{B}%
\in\mathbb{R}^{(k_{2}+1)\times m}$ is the matrix with entries $b_{h,j}$ and
$\mathbf{f}_{t}=(f_{t},\dots,f_{t-k_{2}})^{\prime}$. For $h=0,\dots,k_{2}$ let
$\mathbf{b}_{h}=(b_{h,1},\dots,b_{h,m})^{\prime}$. The term $\boldsymbol{\chi
}_{t} = \mathbf{B}^{\prime}\mathbf{f}_{t}$ is usually called the common part
of the model. We will need the following assumptions.

\begin{condition}
$ $

\begin{itemize}
\item[(a)] $\mathbf{B}\mathbf{B}^{\prime}/m\rightarrow\mathbf{I}_{k_{2}+1}$.

\item[(b)] $\mathbf{e}_{t}$ and $f_{t}$ are second order stationary,
$\mathbb{E}\mathbf{e}_{t}=\mathbf{0}_{m}$ and $\mathbb{E}f_{t}=0$. Let
$\boldsymbol{\Sigma}^{e}(l)$ be the lag $l$ autocovariance matrix of
$\mathbf{e}_{t}$. Then $\lambda_{max}(\boldsymbol{\Sigma}^{e}(0))=O(1)$.

\item[(c)] $\mathbf{e}_{t}$ is uncorrelated with $\mathbf{f}_{t}$ at all leads
and lags.
\end{itemize}

\label{condition:DFM}
\end{condition}

Condition \ref{condition:DFM}(a) is a standardization assumption. It appears,
for example, in \cite{Pena87}. See also \cite{BaiNg}. Conditions
\ref{condition:DFM}(b) and (c) allow for weak cross-sectional correlations in
the idiosyncratic part.

The following theorem shows that the population reconstruction mean squared
error of the ODPC procedure is essentially bounded by the mean variance of the
idiosyncratic part. This can also be interpreted as a sequential limit
asymptotic: first let $T$ go to infinity for fixed $m$ and then let $m$ go to
infinity. See \cite{Connor} for another example of sequential limit
asymptotics. To keep the notation light, the theorem is stated and proved for
the case in which an intercept $\boldsymbol{\alpha}$ is not included in the
definition of the ODPC; the adjustments to include $\boldsymbol{\alpha}$ are straightforward.

\begin{theorem}
\label{theo:dfm} Assume Conditions \ref{condition:eig} and \ref{condition:DFM}
hold. Then as $m\rightarrow\infty$
\[
\frac{1}{m}\text{MSE}_{0}(\mathbf{a}^{\ast},\mathbf{B}^{\ast})\leq\frac{1}%
{m}\sum_{j=1}^{m}\mathbb{E}e_{t,j}^{2}+o(1),\quad\text{for all }%
(\mathbf{a}^{\ast},\mathbf{B}^{\ast})\in\mathcal{I}.
\]

\end{theorem}

The following technical conditions are needed to ensure that the decomposition
$\mathbf{z}_{t}= \boldsymbol{\chi}_{t} + \mathbf{e}_{t}$ is unique. See
Theorem B of \cite{Forni2015}.

\begin{condition}
$ $

\begin{itemize}
\label{condition:DFM:exist}

\item[(a)] For each $m$, $\mathbf{z}_{t}$ is a second order $m$-dimensional
stationary process that has a spectral density.

\item[(b)] Let $\boldsymbol{\Sigma}(0)$ be the covariance matrix of
$\mathbf{z}_{t}$. Let $\lambda_{m,j}^{\mathbf{z}}$ be its $j$-th eigenvalue
and let $\lambda_{j}^{\mathbf{z}} = \sup_{m \in\mathbb{N}} \lambda
_{m,j}^{\mathbf{z}}$. Then $\lambda_{k_2+1}^{\mathbf{z}}=\infty$ and
$\lambda_{k_2+2}^{\mathbf{z}}<\infty$.
\end{itemize}
\end{condition}

If we further assume that the idiosyncratic disturbances are cross-sectionally
uncorrelated, we can prove that the ODPC is able to recover the common part of
the dynamic factor model asymptotically.

\begin{condition}
$ $

\begin{itemize}
\item[(a)] There exists $L>0$ such that $\mathbb{E}z_{t,j}^{2}\leq L$ for all
$j$.

\item[(b)] $\boldsymbol{\Sigma}^{e}(0)$ is a diagonal matrix.
\end{itemize}

\label{condition:common}
\end{condition}

\begin{theorem}
\label{theo:common} Assume Conditions \ref{condition:eig}, \ref{condition:DFM}%
, \ref{condition:DFM:exist} and \ref{condition:common} hold. Then as
$m\rightarrow\infty$
\[
\frac{1}{m}\mathbb{E}\Vert\mathbf{B}^{\prime}\mathbf{f}_{t}-
\widehat{\mathbf{z}}_{t}\Vert^{2}\rightarrow0,
\]
where $\widehat{\mathbf{z}}_{t}=\mathbf{z}^{R}_{t}(\mathbf{a}^{\ast},
\mathbf{B}^{\ast})$, for $(\mathbf{a}^{\ast}, \mathbf{B}^{\ast})\in
\mathcal{I}$.
\end{theorem}

\section{Simulation study}

\label{sec:simu} In this Section, we compare the procedure proposed in this
paper (ODPC) with those of \cite{Forni2005} (FHLR), \cite{Forni2015} (FHLZ)
and \cite{SW2002} (SW) for forecasting multivariate time series.

We took $(T,m)\in\{50,100,200\}\times\{50,100,200\}$. This Monte Carlo design
includes difficult forecasting situations where the ratio $T/m$ is smaller
than one. For each combination of $T$ and $m$, we generated 500 vector time
series with $T+1$ periods of the following models.

\begin{enumerate}
\item[DFM1] The generating model is $z_{t,j}=c(\sin(2\pi j/m)f_{t}+\cos(2\pi
j/m)f_{t-1}+(j/m)f_{t-2}+f_{t-3})+u_{t,j}$ where the $u_{t,j}$ are i.i.d.
standard normal random variables. The factor $f_{t}$ is generated according to
a moving average process $f_{t}=v_{t}+\theta_{1}v_{t-1}+\theta_{2}v_{t-2}$,
where the $v_{t}$ are i.i.d. standard normals. $\theta_{2}$ is generated at
random uniformly on the interval $(-0.7,0.7)$. Then $\theta_{1}$ is generated
at random uniformly on the interval $(0,1-|\theta_{2}|)$. $c$ is chosen so
that the mean empirical variance of the common part is equal to one. This is a
stationary dynamic factor model with four static factors and one dynamic factor.

\item[DFM1AR] This model adds idiosyncratic AR(1) structure to DFM1. For each
$j$, $u_{t,j}$ follows an unit variance AR(1), where at each replication the
autoregression coefficient is generated at random, with uniform distribution
in $(-0.9,0.9)$.

\item[DFM2] The model is now $z_{t,j}=c(\sin(2\pi j/m)f_{t}+\cos(2\pi
j/m)f_{t-1}+(j/m)f_{t-2})+u_{t,j}$ where the $u_{t,j}$ are i.i.d. standard
normal random variables. The factor $f_{t}$ follows an autoregressive model
$f_{t}=1.4f_{t-1}-0.45f_{t-2}+v_{t}$, where the $v_{t}$ are i.i.d. standard
normals. As in model DFM1, $c$ is chosen so that the mean empirical variance
of the common part is equal to one. This is a stationary dynamic factor model
with three static factors and one dynamic factor.

\item[DFM2AR] As in model DFM1 we add idiosyncratic AR(1) structure to DFM2.
Again $u_{t,j}$ follows an unit variance AR(1) and at each replication the
parameter is chosen from an uniform distribution in $(-0.9,0.9)$.

\item[VARMA] We first generate $\mathbf{x}_{t}=\boldsymbol{\Lambda}%
\mathbf{x}_{t-1}+\mathbf{u}_{t}$, where the $\mathbf{u}_{t}$ are i.i.d
standard multivariate normal variables and $\boldsymbol{\Lambda}$ is generated
at random for each replication, as a diagonal matrix where the elements of the
diagonal are independent and generated at random with uniform distribution in
$(-0.9,0.9)$. We then take $\mathbf{z}_{t}=\mathbf{M}\mathbf{x}_{t}$, where
$\mathbf{M}$ is a lower triangular matrix of dimensions $m\times m$ filled
with ones. The $\mathbf{z}_{t}$ follow a stationary VARMA model. Finally, we
standardize the data so that is has empirical mean variance equal to one.
\end{enumerate}

For each estimator, and each combination of $T$ and $m$, using periods
$1,\dots,T$ we compute a forecast of each time series at period $T+1$ and the
corresponding prediction mean squared error (PMSE). For the dynamic factor
models, we only compute a forecast of the common part, that is, we did not
include any forecasting of the idiosyncratic component. We report the average
PMSE over the 500 replications.

The MATLAB code to compute FHLR was obtained from
\url{http://morgana.unimore.it/forni_mario/matlab.htm}. The Bartlett
lag-window size was taken as $[\sqrt{T+1}]$. The MATLAB code to compute FHLZ
was kindly provided by the authors. We used a triangular kernel, with window
size equal to $[(T+1)^{2/3}]$. The maximum order of the singular VARs was
taken to be 5, and the order was chosen using the BIC criterion. The number of
random permutations of the series was taken to be 30. We used our own
implementation of the SW procedure. For FHLR and FHLZ, the procedures were
applied to the data standardized to zero mean and unit variance. At the end,
the forecasts were transformed to the original units. The SW forecast was
obtained by projecting $\mathbf{z}_{T+1}$ on the estimated factors, as in
equation 16 of \cite{Forni2005}. For ODPC, to forecast future values of the
dynamic principal components we use the \texttt{auto.arima} and
\texttt{forecast.arima} functions from the \texttt{forecast} \texttt{R}
package \citep{forecast} to automatically fit (possibly seasonal) ARIMA models
to the dynamic components and obtain their forecasts.

For the factor models we use the known number of factors and lags. For the
ODPC procedure, we take one component with $k_{1}=k_{2}$, equal to the number
of lags in the factors in the generated model. Results for each combination of
$(T,m)$, are shown in Tables \ref{tab:sim:DFM1:MSE} and \ref{tab:sim:DFM2:MSE}.

For the VARMA, model we use three different combinations of number of
components and lags. We have used one, two and five dynamic factors for ODPC,
FHLR and FHLZ and the equivalent, or larger, number of static factors for SW.
Results are shown in Table \ref{tab:sim:VARMA:MSE}. The second rows show, for
ODPC (number of components, $k_{1}$, $k_{2}$), for FHLR (number of dynamic
factors, number of static factors), for FHLZ (number of dynamic factors) and
for SW (number of static factors).

Highlighted in black is the best result. An asterisk indicates that the
difference with the runner up is significant at the 95\% level, taking the
difference of the squared forecasting errors. Table \ref{tab:sim:DFM1:MSE}
shows that for models DFM1 and DFM1AR, with MA factors, the method ODPC always
works better than the competitors, although the differences tend to be small.
The largest difference with respect to the runner-up in both models is around
10\%. Table \ref{tab:sim:DFM2:MSE} shows that for AR factors, models DFM2 and
DFM2AR, the errors of all the methods are smaller, as expected, and ODPC
performs similar to SW and slightly better than the others. However, note that
ODPC makes forecasts with only one dynamic component whereas the SW procedure
requires three or four static factors. In Table \ref{tab:sim:VARMA:MSE} again
ODPC is most often the winner, although the differences with the runner-up are
not large.

In summary, the proposed procedure seems to work well, both for dynamic factor
models and for VARMA models with large common dependency.

\newpage

\begin{table}[th]
\centering
\begin{tabular}
[c]{llllllllll}\hline
&  & DFM1 &  &  &  & DFM1AR &  &  & \\\hline
$T$ & $m$ & ODPC & FHLR & FHLZ & SW & ODPC & FHLR & FHLZ & SW\\\hline
50 & 50 & \textbf{1.460*} & 1.510 & 1.742 & 1.542 & \textbf{1.429} & 1.460 &
1.617 & 1.493\\
& 100 & \textbf{1.411*} & 1.454 & 1.637 & 1.533 & \textbf{1.469} & 1.495 &
1.595 & 1.535\\
& 200 & \textbf{1.436} & 1.446 & 1.643 & 1.527 & \textbf{1.413*} & 1.464 &
1.575 & 1.509\\
100 & 50 & \textbf{1.327*} & 1.460 & 1.569 & 1.489 & \textbf{1.341*} & 1.453 &
1.543 & 1.479\\
& 100 & \textbf{1.294*} & 1.367 & 1.512 & 1.439 & \textbf{1.306*} & 1.411 &
1.509 & 1.452\\
& 200 & \textbf{1.262*} & 1.330 & 1.494 & 1.378 & \textbf{1.310*} & 1.377 &
1.481 & 1.418\\
200 & 50 & \textbf{1.275*} & 1.373 & 1.484 & 1.413 & \textbf{1.271*} & 1.397 &
1.457 & 1.396\\
& 100 & \textbf{1.216*} & 1.313 & 1.403 & 1.365 & \textbf{1.283*} & 1.384 &
1.489 & 1.412\\
& 200 & \textbf{1.232*} & 1.302 & 1.448 & 1.360 & \textbf{1.240*} & 1.316 &
1.413 & 1.370\\\hline
\end{tabular}
\caption{Means of the 1-step ahead PMSEs of ODPC, FHLR, FHLZ and SW for models
DFM1 and DFM1AR.}%
\label{tab:sim:DFM1:MSE}%
\end{table}

\begin{table}[th]
\centering
\begin{tabular}
[c]{llllllllll}\hline
&  & DFM2 &  &  &  & DFM2AR &  &  & \\\hline
$T$ & $m$ & ODPC & FHLR & FHLZ & SW & ODPC & FHLR & FHLZ & SW\\\hline
50 & 50 & 1.268 & 1.286 & 1.712 & \textbf{1.243* } & \textbf{1.205} & 1.269 &
1.757 & 1.208\\
& 100 & \textbf{1.217} & 1.277 & 1.594 & 1.219 & 1.199 & 1.238 & 1.582 &
\textbf{1.184*}\\
& 200 & \textbf{1.185*} & 1.271 & 1.525 & 1.210 & \textbf{1.165} & 1.243 &
1.520 & 1.174\\
100 & 50 & 1.149 & 1.146 & 1.357 & \textbf{1.128*} & 1.132 & 1.147 & 1.375 &
\textbf{1.123}\\
& 100 & \textbf{1.119} & 1.151 & 1.322 & 1.124 & 1.110 & 1.136 & 1.309 &
\textbf{1.110}\\
& 200 & \textbf{1.103*} & 1.136 & 1.265 & 1.112 & \textbf{1.097} & 1.127 &
1.280 & 1.100\\
200 & 50 & 1.110 & 1.095 & 1.202 & \textbf{1.086*} & 1.120 & 1.113 & 1.260 &
\textbf{1.097*}\\
& 100 & 1.092 & 1.087 & 1.174 & \textbf{1.078*} & 1.083 & 1.089 & 1.201 &
\textbf{1.076*}\\
& 200 & \textbf{1.073*} & 1.089 & 1.174 & 1.079 & \textbf{1.063} & 1.082 &
1.188 & 1.066\\\hline
\end{tabular}
\caption{Means of the 1-step ahead PMSEs of ODPC, FHLR, FHLZ and SW for models
DFM2 and DFM2AR.}%
\label{tab:sim:DFM2:MSE}%
\end{table}

\begin{sidewaystable}[ht]
\centering
{\small
\begin{tabular}{llllllllllllll}
\hline
$T$ & $m$ & ODPC & FHLR & FHLZ & SW & ODPC & FHLR & FHLZ & SW & ODPC & FHLR & FHLZ & SW\\  \hline
&     & (1, 1, 1) & (1, 2) & (1)& (2) & (2, 1, 1) & (2, 6)& (2) & (6) & (5, 1, 1) & (5, 10) & (5)	& (10) \\
\hline
50 & 50 &\textbf{0.997} & 1.024 & 1.042 & 1.042 & 1.010 & 1.017 &\textbf{0.978} & 1.052 & 1.020 & 0.995 & \textbf{0.972} & 1.029 \\
& 100 & 1.001 & \textbf{0.992} & 1.025 & 1.007 & 1.010 & 1.027 & \textbf{0.972} & 1.049 & 1.041 & 1.026 & \textbf{0.972*} & 1.060 \\
& 200 & \textbf{0.999} & 1.047 & 1.015 & 1.050 & 1.010 & 1.034 &\textbf{0.974} & 1.063 & 1.017 & 1.035 & \textbf{0.996} & 1.079 \\
100 & 50 & \textbf{0.863} & 0.909 & 0.891 & 0.926 & \textbf{0.854} & 0.858 & 0.862 & 0.900 & 0.854 & \textbf{0.835} & 0.872 & 0.866 \\
& 100 & \textbf{0.926*} & 1.000 & 0.976 & 1.021 & \textbf{0.918} & 0.986 & 0.933 & 1.013 & \textbf{0.920} & 0.956 & 0.959 & 0.991 \\
& 200 & \textbf{0.938*} & 1.038 & 0.993 & 1.021 & \textbf{0.933} & 1.041 & 0.943 & 1.071 & \textbf{0.933*} & 1.031 & 1.006 & 1.044 \\
200 & 50 & \textbf{0.856} & 0.898 & 0.905 & 0.906 & 0.844 & 0.840 & \textbf{0.833} & 0.880 & 0.836 & \textbf{0.811} & 0.860 & 0.851 \\
& 100 & \textbf{0.826*} & 0.951 & 0.884 & 0.963 & \textbf{0.818} & 0.932 & 0.832 & 0.954 & \textbf{0.812*} & 0.912 & 0.887 & 0.921 \\
& 200 & \textbf{0.893*} & 0.996 & 0.950 & 0.997 & \textbf{0.878} & 0.969 & 0.882 & 0.995 &\textbf{0.873*} & 0.973 & 0.919 & 0.976 \\
\hline
\end{tabular}
\caption{Means of the 1-step ahead PMSEs of ODPC, FHLR, FHLZ and SW for the VARMA model.}
\label{tab:sim:VARMA:MSE}
}
\end{sidewaystable}

\FloatBarrier

\section{An empirical example}

\label{sec:real_data} In this section, we compare the forecasting performances
of the procedures considered in the previous section when applied to a panel
of real macroeconomic variables. The data set, downloaded from
\url{https://research.stlouisfed.org/econ/mccracken/fred-databases/}, consists
of several key monthly macroeconomic variables for the US economy used by
Stock and Watson (2002) . A full description of the data can be found in the
website. See also \cite{FRED}. The data was corrected for outliers and
transformed to stationarity using the MATLAB script provided in the
aforementioned website. We kept only periods from January 1960 to February
2014 and removed series with missing data, resulting in a balanced panel with
$T=650$ observations on $m=94$ series. Let $\mathbf{Z}=\lbrace z_{t,j}
\rbrace$ be the resulting panel.

Following \cite{FRED} we used four series as target variables for forecasting:
CLAIMSx, initial jobless claims, S\&P: indust, the S\&P Industrial Index, M2REAL,  Real M2 Money Stock and INDPRO, Industrial Production Index, all in log levels. Let $j$ be the
index of any of the target variables. Since the targets are transformed by
taking first differences of the logarithm, the target at time $t+h$ is
$z_{t+1,j}+\dots+z_{t+h,j}$.

We considered one and two years forecast horizons, $h=12, 24$, and selected the most recent
out of sample forecast period of \cite{FRED}, the one going from 2008:01 to
2014:12. Thus, we fit the four procedures discussed in the previous Monte
Carlo section using sample periods $1,\dots,(T-h-t)$ for each $t=0,\dots,83$
to predict $T-t$ that is, we use a rolling seven year window, covering the
recovery of the US economy. We compute the $h$-steps ahead forecasts of the
whole panel, and then compare the predicted values of the target variable with
the actual values. Let
\[
E_{t,h}=(z_{T-t-h+1,6}+\dots+z_{T-t,6}-(\widehat{z}_{T-t-h+1,6|T-h-t}%
+\dots+\widehat{z}_{T-t,6|T-h-t}))
\]
be the forecasting error of period $T-t$ using information up to period
$T-h-t$, for $t=0,\dots,83$. We measure the performance of each procedure by
\[
\left(  \frac{1}{84}\sum\limits_{t=0}^{83}E_{t,h}^{2}\right)  ^{1/2}.
\]
As in \cite{FRED} we compare the forecasts obtained by the different procedures
using the first factor with lags. Thus, we did not try  to forecast the
idiosyncratic part, as our objective is to compare the performance of the
methods in forecasting the common component in the series. We computed: ODPC
with one component and up to three lags, FHLR with one dynamic factor and up
to four static factors (this amount to assuming that the dynamic factor is
loaded with up to three lags), FHLZ with one dynamic factor and SW with one
static factor and up to three of its lags. We also computed a one-dimensional
SARIMA forecast, by automatically fitting a SARIMA model using the
\texttt{auto.arima} function from the \texttt{forecast} \texttt{R} package,
using the default settings.

We report the root means squared forecasting errors relative to those of the one-dimensional SARIMA forecast. Results are shown in Table \ref{tab:real:RMSE}. For a one year horizon, $h=12$, in half of the four series the best forecast is obtained with ODPC, that reduces the univariate forecast errors by 12.1\% in CLAIMSx, and by 2.9\% in S\&P Indust. The largest reduction of error with respect to the univariate forecast is in M2REAL, where all the procedures reduce the forecast error between 14,9\% and 11,9\% and the winner in this case is SW. For INDPRO the maximum error reduction is 8,5\% and is obtained by FHLR. For the two year horizon, $h=24$, for three of the four series, the best forecast is obtained with ODPC, reducing the univariate forecast errors by 28.8\% in CLAIMSx, 10.3\% in S\&P indust and 13.2\% in INDPRO. For M2REAL the best forecast is again obtained with SW, achieving a 19.8\% reduction in error with respect to the univariate method.
The conclusion is that the results of the four precedures in this data set are similar with a small advantage of ODPC.

\begin{table}[ht]
\centering
\begin{tabular}{lllll}
  \hline
 & CLAIMSx & S\&P Indust & M2REAL & INDPRO \\ 
  \hline
h = 12 &          &           &         & \\
  \hline
ODPC 1 & \textbf{0.879} & \textbf{0.971} & 0.901 & 1.001 \\ 
  ODPC 2 & 0.894 & 0.972 & 0.891 & 1.025 \\ 
  ODPC 3 & 0.913 & 0.981 & 0.869 & 1.041 \\ 
  FHLR 1 & 1.003 & 1.019 & 0.881 & 0.944 \\ 
  FHLR 2 & 0.904 & 0.997 & 0.897 & 0.931 \\ 
  FHLR 3 & 0.920 & 1.006 & 0.903 & \textbf{0.915} \\ 
  FHLZ  & 0.994 & 1.003 & 0.868 & 0.960 \\ 
  SW 1 & 1.004 & 1.006 & 0.853 & 1.033 \\ 
  SW 2 & 1.018 & 1.018 & 0.852 & 1.046 \\ 
  SW 3 & 1.028 & 1.029 & \textbf{0.851} & 1.058 \\ 
  \hline
  h = 24 &          &           &         & \\ \hline
  ODPC 1 &\textbf{0.712} & 0.900 & 0.932 & \textbf{0.868} \\ 
  ODPC 2 & 0.722 & \textbf{0.897} & 0.930 & 0.886 \\ 
  ODPC 3 & 0.729 & 0.908 & 0.897 & 0.903 \\ 
  FHLR 1 & 1.006 & 1.003 & 0.848 & 0.933 \\ 
  FHLR 2 & 0.877 & 0.974 & 0.855 & 0.888 \\ 
  FHLR 3 & 0.892 & 0.982 & 0.849 & 0.883 \\ 
  FHLZ & 1.002 & 0.994 & 0.855 & 0.938 \\ 
  SW 1 & 1.086 & 1.056 & 0.803 & 1.035 \\ 
  SW 2 & 1.101 & 1.067 & \textbf{0.802} & 1.047 \\ 
  SW 3 & 1.116 & 1.078 & 0.802 & 1.058 \\ 
   \hline
\end{tabular}
\caption{RMSE forecasting errors for different number of lags,
                  relative to the RMSE of the one dimensional SARIMA forecast.} 
\label{tab:real:RMSE}
\end{table}

\FloatBarrier

\section{Choosing the number of components and lags}
\label{sec:lags}
In practice, the number of components and lags needs to be chosen. To simplify
the notation, assume that for each component $k_{1}^{i}=k_{2}^{i}$, that is,
the number of lags of $\mathbf{z}_{t}$ used to define the dynamic principal
component and the number of lags of $\widehat{f}_{t}$ used to reconstruct the
original series are the same.

One possible approach is to minimize the cross-validated forecasting error in
a stepwise fashion. Choose a maximum number of lags $K_{max},$ and, starting
with one component, search for the value $k^{\ast}$ among $0,\dots,K_{max}$
that gives the minimum cross-validated forecasting error. Then, fix the first
component computed with $k^{\ast}$ lags and repeat the procedure with the
second component. If the optimal cross-validated forecasting error using the
two components is larger than the one using only one component, stop;
otherwise add a third component and proceed as before.

The same stepwise approach could be applied to minimize an information
criterion. This would reduce the computational burden significantly. The
following BIC type criterion could be used. Suppose we have computed $q$
dynamic principal components, each with $k_{1}^{i}=k_{2}^{i}=k^{i}$ lags. Let
$\widehat{y}_{t,j}=\widehat{\alpha}_{j}+\sum_{h=0}^{k^{q}}\widehat{\beta
}_{h,j}\widehat{f}_{t-h},t=2\sum_{i=1}^{q}k^{i}+1,\dots,T$ be the
reconstruction obtained, where $y_{t,j}=z_{t,j}$ for the first component and
will be equal to the residuals from the fit with the previous components
otherwise. Let $r_{t,j}=y_{t,j}-\widehat{y}_{t,j}$ be the residuals,
$\mathbf{R}_{q}$ be the corresponding matrix of residuals and
$\boldsymbol{\Sigma}_{q}=(\mathbf{R}_{q}^{\prime}\mathbf{R}_{q})/\left(
T-2\sum_{i=1}^{q}k^{i}\right)  $. Then for each $q$ choose the value $k^{\ast
}$ among $0,\dots,K_{max}$ that minimizes
\[
\text{BIC}_{k}=\left(  T-2\sum_{i=1}^{q}k^{i}\right)  \log\left(
\text{trace}(\boldsymbol{\Sigma}_{q})\right)  +m(2k+3)\log\left(
T-2\sum_{i=1}^{q}k^{i}\right)  .
\]

The performance of these alternatives will be the subject of further research.

\section{Conclusions and possible extensions}

\label{sec:conclu}

We have presented a new procedure for the dimension reduction of multivariate
time series. The main advantages with respect to other alternatives are that:
in the spirit of principal component analysis, it is not based on assuming any
particular model (parametric or not) for the data, but rather on finding
linear combinations of the observations with optimal reconstruction
properties; not being based on both lags and leads of the data, it is useful
for forecasting large sets of time series.

Moreover, the proposed procedure can be generalized in several directions.
First, since the MSE criterion used in the is paper is not robust, it can
substituted for the minimization of a robust scale. This can be achieved in a
similar way as in \cite{PenaYohai2016}.\ A simpler way to obtain robustness
would be to substitute the alternating least squares regressions by robust
regression estimators, for example MM-estimators \citep{MM87}. Second, to deal
with very large number of variables the estimation algorithm can be
regularized. For example, in each of the steps, the alternating least squares
regressions may be replaced by a regularized regressions, using, for example,
a Lasso procedure. This method will allow for a different number of lags in
different variables. Both modifications, for robustness and regularization,
can be combined using in each step a robust lasso procedure (see for example
\cite{SY17}). All these extensions require further research.

\section{Appendix}
\label{sec:appen}
This section includes the proofs of all the results stated in the paper.

\begin{lemma}
\label{lemma:conv_unif_cov}
\begin{align*}
\sup_{\Vert\mathbf{a}\Vert=1} \left\Vert \frac{\mathbf{F}_{k_{1}, k_{2}%
}^{\prime}\mathbf{F}_{k_{1}, k_{2}}}{T-(k_{1} + k_{2})} - \mathcal{S}%
(\mathbf{a})\right\Vert _{F} \overset{a.s.}{\rightarrow} 0.
\end{align*}

\end{lemma}

\begin{proof}
[Proof of Lemma \ref{lemma:conv_unif_cov}]Fix $\mathbf{a}$ with $\Vert
\mathbf{a}\Vert=1$. Then
\begin{align*}
&  \frac{\mathbf{F}_{k_{1}, k_{2}}^{\prime}\mathbf{F}_{k_{1}, k_{2}}}{T-(k_{1}
+ k_{2})}\\
&  = \frac{1}{T-(k_{1} + k_{2})}
\begin{pmatrix}
\mathbf{1}_{T-(k_{1} + k_{2})}^{\prime}\\
\mathbf{a}^{\prime} \mathbf{Z}_{k_{1} + k_{2},0}^{\prime}\\
\mathbf{a}^{\prime} \mathbf{Z}_{k_{1} + k_{2}-1,0}^{\prime}\\
\vdots\\
\mathbf{a}^{\prime} \mathbf{Z}_{k_{1},0}^{\prime}%
\end{pmatrix}
\begin{pmatrix}
\mathbf{1}_{T-(k_{1} + k_{2})} & \mathbf{Z}_{k_{1} + k_{2},0}\mathbf{a} &
\mathbf{Z}_{k_{1} + k_{2}-1,0}\mathbf{a} & \dots & \mathbf{Z}_{k_{1}%
,0}\mathbf{a}%
\end{pmatrix}
\end{align*}
Fix $k_{1}\leq i,j \leq k_{1} + k_{2}$. Then
\begin{align*}
\mathbf{Z}_{i,0}^{\prime}\mathbf{Z}_{j,0}  &  =
\begin{pmatrix}
\mathbf{Z}_{i}^{\prime}\\
\mathbf{Z}_{i-1}^{\prime}\\
\vdots\\
\mathbf{Z}_{i-k_{1}}^{\prime}\\
\end{pmatrix}
\begin{pmatrix}
\mathbf{Z}_{j} & \mathbf{Z}_{j-1} & \dots & \mathbf{Z}_{j-k_{1}}%
\end{pmatrix}
\\
&  =
\begin{pmatrix}
\mathbf{Z}_{i}^{\prime} \mathbf{Z}_{j} & \mathbf{Z}_{i}^{\prime}
\mathbf{Z}_{j-1} & \dots & \mathbf{Z}_{i}^{\prime} \mathbf{Z}_{j-k_{1}}\\
\vdots & \vdots &  & \vdots\\
\mathbf{Z}_{i-k_{1}}^{\prime} \mathbf{Z}_{j} & \mathbf{Z}_{i-k_{1}}^{\prime}
\mathbf{Z}_{j-1} & \dots & \mathbf{Z}_{i-k_{1}}^{\prime} \mathbf{Z}_{j-k_{1}}%
\end{pmatrix}
.
\end{align*}
Note that
\begin{align*}
\mathbf{Z}_{i}^{\prime}\mathbf{Z}_{j} =\sum\limits_{r=1}^{T-(k_{1} + k_{2})}
\mathbf{z}_{i+r}\mathbf{z}_{j+r}^{\prime}.
\end{align*}
By the Ergodic Theorem
\begin{align*}
\frac{\mathbf{Z}_{i}^{\prime}\mathbf{Z}_{j}}{T-(k_{1} + k_{2})}
\overset{a.s.}{\rightarrow} \boldsymbol{\Sigma}( i-j ).
\end{align*}
Hence
\begin{align*}
\frac{\mathbf{Z}_{i,0}^{\prime}\mathbf{Z}_{j,0}}{T-(k_{1} + k_{2})}
\overset{a.s.}{\rightarrow}
\begin{pmatrix}
\boldsymbol{\Sigma}( i-j) & \boldsymbol{\Sigma}( i-j+1 ) & \dots &
\boldsymbol{\Sigma}( i-j+k_{1} )\\
\vdots & \vdots &  & \vdots\\
\boldsymbol{\Sigma}( i-j-k_{1}) & \boldsymbol{\Sigma}(i-j-k_{1}+1 ) & \dots &
\boldsymbol{\Sigma}( i-j )
\end{pmatrix}
= \mathbf{V}(i-j).
\end{align*}
On the other hand
\begin{align*}
\frac{1}{T-(k_{1} + k_{2})}\mathbf{a}^{\prime}\mathbf{Z}_{i,0}^{\prime
}\mathbf{1}_{T-(k_{1} + k_{2})}  &  = \frac{1}{T-(k_{1} + k_{2})}
\sum\limits_{h=0}^{k_{1}} \mathbf{a}_{h}^{\prime} \mathbf{Z}^{\prime}_{i-h}
\mathbf{1}_{T-(k_{1} + k_{2})}\\
&  = \frac{1}{T-(k_{1} + k_{2})} \sum\limits_{h=0}^{k_{1}} \mathbf{a}%
_{h}^{\prime}
\begin{pmatrix}
\sum\limits_{r=1}^{T-(k_{1} + k_{2})} z_{i-h+r, 1}\\
\vdots\\
\sum\limits_{r=1}^{T-(k_{1} + k_{2})} z_{i-h+r, m}%
\end{pmatrix}
\overset{a.s.}{\rightarrow} 0,
\end{align*}
by the Ergodic Theorem and since $\mathbb{E}\mathbf{z}_{t}=0$ by assumption.
We have shown that
\begin{align*}
\frac{\mathbf{F}_{k_{1}, k_{2}}^{\prime}\mathbf{F}_{k_{1}, k_{2}}}{T-(k_{1} +
k_{2})} \overset{a.s.}{\rightarrow}
\begin{pmatrix}
1 & 0 & 0 & \dots & 0\\
0 & \mathbf{a}^{\prime} \mathbf{V}(0)\mathbf{a} & \mathbf{a}^{\prime}
\mathbf{V}(1)\mathbf{a} & \dots & \mathbf{a}^{\prime} \mathbf{V}%
(k_{2})\mathbf{a}\\
\vdots & \vdots & \vdots &  & \vdots\\
0 & \mathbf{a}^{\prime} \mathbf{V}(-k_{2})\mathbf{a} & \mathbf{a}^{\prime}
\mathbf{V}(-k_{2}+1)\mathbf{a} & \dots & \mathbf{a}^{\prime} \mathbf{V}%
(0)\mathbf{a}%
\end{pmatrix}
= \mathcal{S}(\mathbf{a}).
\end{align*}

To prove that the convergence holds uniformly, it suffices to show that, for
$k_{1} \leq i,j\leq k_{1} + k_{2}$,
\begin{align*}
&  \sup_{\Vert\mathbf{a}\Vert=1} \left\Vert \frac{1}{T-(k_{1} + k_{2})}%
\sum\limits_{h=0}^{k_{1}} \mathbf{a}_{h}^{\prime} \mathbf{Z}^{\prime}_{i-h}
\mathbf{1}_{T-(k_{1} + k_{2})}\right\Vert \overset{a.s.}{\rightarrow} 0 \text{
and }\\
&  \sup_{\Vert\mathbf{a}\Vert=1} \left\vert \frac{\mathbf{a}^{\prime}
\mathbf{Z}_{i,0}^{\prime}\mathbf{Z}_{j,0}\mathbf{a} }{T-(k_{1} + k_{2})} -
\mathbf{a}^{\prime} \mathbf{V}(i-j)\mathbf{a}\right\vert
\overset{a.s.}{\rightarrow} 0
\end{align*}
The first assertion follows immediately from the Ergodic Theorem. It is easy
to show that
\begin{align*}
\frac{\mathbf{a}^{\prime} \mathbf{Z}_{i,0}^{\prime}\mathbf{Z}_{j,0}\mathbf{a}
}{T-(k_{1} + k_{2})}=\frac{\sum\limits_{h=0}^{k_{1}}\sum\limits_{r=0}^{k_{1}}
\mathbf{a}_{r}^{\prime} \mathbf{Z}^{\prime}_{i-r} \mathbf{Z}_{j-h}
\mathbf{a}_{h}}{T-(k_{1} + k_{2})}.
\end{align*}
Note that for any $\mathbf{v}, \mathbf{w}\in\mathbb{R}^{m}$
\begin{align*}
\mathbf{v}^{\prime}\mathbf{Z}_{i}^{\prime}\mathbf{Z}_{j}\mathbf{w}
=\mathbf{v}^{\prime}(\sum\limits_{r=1}^{T-(k_{1} + k_{2})} \mathbf{z}%
_{i+r}\mathbf{z}_{j+r}^{\prime})\mathbf{w}.
\end{align*}
Thus, to prove the lemma it will be enough to prove that
\begin{align*}
\sup_{\Vert\mathbf{v}\Vert\leq1, \Vert\mathbf{w}\Vert\leq1 } \left\vert
\frac{\mathbf{v}^{\prime}(\sum\limits_{r=1}^{T-(k_{1} + k_{2})} \mathbf{z}%
_{i+r}\mathbf{z}_{j+r}^{\prime})\mathbf{w}}{T-(k_{1} + k_{2})} -
\mathbf{v}^{\prime}\boldsymbol{\Sigma}( i -j) \mathbf{w}\right\vert
\overset{a.s.}{\rightarrow} 0. \label{eq:conv_unif_cov}%
\end{align*}
This follows immediately from
\begin{align*}
\sup_{\Vert\mathbf{v}\Vert\leq1 \Vert\mathbf{w}\Vert\leq1 } \left\vert
\frac{\mathbf{v}^{\prime}(\sum\limits_{r=1}^{T-(k_{1} + k_{2})} \mathbf{z}%
_{i+r}\mathbf{z}_{j+r}^{\prime})\mathbf{w}}{T-(k_{1} + k_{2})} -
\mathbf{v}^{\prime}\boldsymbol{\Sigma}( i -j) \mathbf{w}\right\vert
\leq\left\Vert \frac{\sum\limits_{r=1}^{T-(k_{1} + k_{2})} \mathbf{z}%
_{i+r}\mathbf{z}_{j+r}^{\prime}}{T-(k_{1} + k_{2})} - \boldsymbol{\Sigma}( i
-j)\right\Vert
\end{align*}
and the Ergodic Theorem.
\end{proof}

\begin{lemma}
\label{lemma:conv_min_eig}
\begin{align*}
\liminf\limits_{T}\inf_{\Vert\mathbf{a} \Vert=1} \lambda_{min}\left(
\frac{\mathbf{F}_{k_{1}, k_{2}}^{\prime}\mathbf{F}_{k_{1}, k_{2}}}{T-(k_{1} +
k_{2})}\right)  \geq\inf_{\Vert\mathbf{a}\Vert=1} \lambda_{min}(\mathcal{S}%
(\mathbf{a})),
\end{align*}
with probability one.
\end{lemma}

\begin{proof}
[Proof of Lemma \ref{lemma:conv_min_eig}]It suffices to show that
\begin{align*}
\sup_{\Vert\mathbf{a} \Vert=1}\left\vert \lambda_{min}\left(  \frac
{\mathbf{F}_{k_{1}, k_{2}}^{\prime}\mathbf{F}_{k_{1}, k_{2}}}{T-(k_{1} +
k_{2})}\right)  - \lambda_{min}(\mathcal{S}(\mathbf{a})) \right\vert
\overset{a.s.}{\rightarrow} 0
\end{align*}
and this follows from Theorem 3.3.16 of \cite{Matrix-Analysis} and Lemma
\ref{lemma:conv_unif_cov}.
\end{proof}

To ease the notation, we will note
\[
g(M) = \left(  M \left(  \inf_{\Vert\mathbf{a} \Vert= 1} \lambda_{min}\left(
\mathcal{S}(\mathbf{a}) \right)  \right)  ^{1/2} - \left(  \mathbb{E}%
\Vert\mathbf{z}\Vert^{2}\right)  ^{1/2}\right)  ^{2}.
\]
If Condition \ref{condition:eig} holds, clearly $g(M) \rightarrow+\infty$ when
$M \rightarrow+\infty$. The following Lemma is a key result.

\begin{lemma}
\label{lemma:liminf_MSE} Assume Condition \ref{condition:eig} holds. Then if
$M>(\mathbb{E}\Vert\mathbf{z} \Vert^{2}/ \inf_{\Vert\mathbf{a} \Vert= 1}
\lambda_{min}\left(  \mathcal{S}(\mathbf{a}) \right)  )^{1/2}$, with
probability 1
\begin{align*}
\lim\inf\limits_{T}\inf_{\Vert\mathbf{a} \Vert= 1, \Vert\mathbf{D}\Vert
_{F}\geq M} \text{MSE}(\mathbf{a},\mathbf{D}) \geq g(M).
\end{align*}

\end{lemma}

\begin{proof}
[Proof of Lemma \ref{lemma:liminf_MSE}]Note that the triangle inequality
implies that
\begin{align*}
\text{MSE}(\mathbf{a},\mathbf{D})^{1/2}\geq\frac{\Vert\mathbf{F}_{k_{1},
k_{2}}\mathbf{D}\Vert_{F} -\Vert\mathbf{Z}_{2k}\Vert_{F} }{\sqrt{T-(k_{1} +
k_{2})}}=\left\Vert \frac{\mathbf{F}_{k_{1}, k_{2}}}{\sqrt{T-(k_{1} + k_{2})}%
}\mathbf{D}\right\Vert _{F} -\left\Vert \frac{\mathbf{Z}_{k_{1} + k_{2}}%
}{\sqrt{T-(k_{1} + k_{2})}}\right\Vert _{F}.
\end{align*}
We will bound the right hand side of the last inequality. It follows from the
Ergodic Theorem that
\begin{align*}
\left\Vert \frac{\mathbf{Z}_{k_{1} + k_{2}}}{\sqrt{T-(k_{1} + k_{2})}%
}\right\Vert _{F}^{2} = \frac{1}{T-(k_{1} + k_{2})} \sum\limits_{t=(k_{1} +
k_{2})+1}^{T}\Vert\mathbf{z}_{t}\Vert^{2} \overset{a.s.}{\rightarrow}
\mathbb{E}\Vert\mathbf{z}\Vert^{2}.
\end{align*}
On the other hand
\begin{align*}
\inf_{\Vert\mathbf{a} \Vert= 1, \Vert\mathbf{D}\Vert_{F}\geq M} \left\Vert
\frac{\mathbf{F}_{k_{1}, k_{2}}}{\sqrt{T-(k_{1} + k_{2})}}\mathbf{D}%
\right\Vert _{F}\geq M \inf_{\Vert\mathbf{a} \Vert= 1} \inf_{\Vert
\mathbf{D}\Vert_{F} = 1} \left\Vert \frac{\mathbf{F}_{k_{1}, k_{2}}}%
{\sqrt{T-(k_{1} + k_{2})}}\mathbf{D}\right\Vert _{F}.
\end{align*}
Note that
\begin{align*}
\inf_{\Vert\mathbf{D}\Vert_{F} = 1} \left\Vert \frac{\mathbf{F}_{k_{1}, k_{2}%
}}{\sqrt{T-(k_{1} + k_{2})}}\mathbf{D}\right\Vert _{F}  &  = \inf
_{\Vert\mathbf{D}\Vert_{F} = 1} \left\Vert vec\left(  \frac{\mathbf{F}_{k_{1},
k_{2}}}{\sqrt{T-(k_{1} + k_{2})}}\mathbf{D} \right)  \right\Vert \\
&  =\inf_{\Vert\mathbf{D}\Vert_{F} = 1} \left\Vert \left(  \mathbf{I}%
_{m}\otimes\frac{\mathbf{F}_{k_{1}, k_{2}}}{\sqrt{T-(k_{1} + k_{2})}}\right)
vec(\mathbf{D})\right\Vert \\
&  =\inf_{\Vert\mathbf{d}\Vert= 1} \left\Vert \left(  \mathbf{I}_{m}%
\otimes\frac{\mathbf{F}_{k_{1}, k_{2}}}{\sqrt{T-(k_{1} + k_{2})}} \right)
\mathbf{d}\right\Vert \\
&  = \lambda_{min}^{1/2} \left(  \left(  \mathbf{I}_{m}\otimes\frac
{\mathbf{F}_{k_{1}, k_{2}}^{\prime}}{\sqrt{T-(k_{1} + k_{2})}} \right)
\left(  \mathbf{I}_{m}\otimes\frac{\mathbf{F}_{k_{1}, k_{2}}}{\sqrt{T-(k_{1} +
k_{2})}} \right)  \right) \\
&  = \lambda_{min}^{1/2}\left(  \mathbf{I}_{m}\otimes\frac{\mathbf{F}_{k_{1},
k_{2}}^{\prime}\mathbf{F}_{k_{1}, k_{2}}}{T-(k_{1} + k_{2})} \right)  =
\lambda_{min}^{1/2}\left(  \frac{\mathbf{F}_{k_{1}, k_{2}}^{\prime}%
\mathbf{F}_{k_{1}, k_{2}}}{T-(k_{1} + k_{2})} \right)  .
\end{align*}
Hence, by Lemma \ref{lemma:conv_min_eig}
\begin{align*}
\liminf_{T}\inf_{\Vert\mathbf{a} \Vert= 1, \Vert\mathbf{D}\Vert_{F}\geq M}
\left\Vert \frac{\mathbf{F}_{k_{1}, k_{2}}}{\sqrt{T-(k_{1} + k_{2})}%
}\mathbf{D}\right\Vert _{F} \geq M \left(  \inf_{\Vert\mathbf{a} \Vert= 1}
\lambda_{min}\left(  \mathcal{S}(\mathbf{a}) \right)  \right)  ^{1/2}.
\end{align*}
It follows that if $M>(\mathbb{E}\Vert\mathbf{z} \Vert^{2}/ \inf
_{\Vert\mathbf{a} \Vert= 1} \lambda_{min}\left(  \mathcal{S}(\mathbf{a})
\right)  )^{1/2}$
\begin{align*}
\lim\inf\limits_{T}\inf_{\Vert\mathbf{a} \Vert= 1, \Vert\mathbf{D}\Vert
_{F}\geq M} \text{MSE}(\mathbf{a},\mathbf{D}) \geq g(M).
\end{align*}

\end{proof}

Let $\mathbb{P}_{T}$ be the empirical probability measure that places mass
$1/(T-(k_{1} + k_{2}))$ at $\mathbf{y}_{1}=(\mathbf{z}_{1},\dots
,\mathbf{z}_{(k_{1}+k_{2})+1}),\dots, \mathbf{y}_{T-(k_{1} + k_{2}%
)}=(\mathbf{z}_{T-(k_{1} + k_{2})},\dots,\mathbf{z}_{T})$. The process
$(\mathbf{y}_{t})_{t}$ is strictly stationary and ergodic. Let $L_{\mathbf{a}%
,\mathbf{D}}(\mathbf{y}_{t})= \Vert\mathbf{z}_{t+(k_{1}+k_{2})}
-\widehat{\mathbf{z}}_{t+(k_{1}+k_{2})}\Vert^{2}$. It follows that
\begin{align*}
\text{MSE}(\mathbf{a},\mathbf{D})=\frac{1}{T-(k_{1} + k_{2})}\sum
\limits_{t=(k_{1}+k_{2})+1}^{T}\Vert\mathbf{z}_{t} -\widehat{\mathbf{z}}%
_{t}\Vert^{2} = \mathbb{P}_{T}L_{\mathbf{a},\mathbf{D}}.
\end{align*}

\begin{lemma}
\label{lemma:conv_unif} For each $M>0$
\begin{align*}
\sup_{\Vert\mathbf{a} \Vert= 1, \Vert\mathbf{D}\Vert_{F}\leq M} \vert
\mathbb{P}_{T}L_{\mathbf{a},\mathbf{D}} - \mathbb{P}L_{\mathbf{a},\mathbf{D}}
\vert\overset{a.s.}{\rightarrow} 0.
\end{align*}

\end{lemma}

\begin{proof}
[Proof of Lemma \ref{lemma:conv_unif}]Let
\begin{align*}
\mathcal{L}=\left\lbrace L_{\mathbf{a},\mathbf{D}}(\cdot): \mathbf{a}
\in\mathbb{R}^{m(k_{1}+1)}, \Vert\mathbf{a} \Vert= 1, \mathbf{D}\in
\mathbb{R}^{(k_{2}+2)\times m}, \Vert\mathbf{D}\Vert_{F}\leq M \right\rbrace .
\end{align*}
$\mathcal{L}$ is VC-major, since it is formed by polynomials of bounded
degree. It has an integrable envelope, since $\mathbb{E}\Vert\mathbf{z}%
\Vert^{2}<+\infty$. Moreover, if we take $\mathcal{L}^{0}$ to be the subset of
$\mathcal{L}$ formed by taking only $\mathbf{a}\in\mathbb{Q}^{m(k_{1}+1)}$ and
$\mathbf{D}\in\mathbb{Q}^{m\times(k_{2}+2)}$ it follows that: $\mathcal{L}%
^{0}$ is countable, and each element of $\mathcal{L}$ is the pointwise limit
of elements of $\mathcal{L}^{0}$. Then the lemma follows from Proposition 1 of
\cite{VCergod}.
\end{proof}

\begin{lemma}
\label{lemma:B_bounded} Assume Condition \ref{condition:eig} holds. Let
$(\widetilde{\mathbf{a}}, \widetilde{\mathbf{D}})$ be such that $\Vert
\widetilde{\mathbf{a}} \Vert=1$ and $\text{MSE}(\widetilde{\mathbf{a}},
\widetilde{\mathbf{D}}) \leq\text{MSE}(\widetilde{\mathbf{a}}, \mathbf{0})$
for all $T$. Fix $M_{0}$ such that
\[
M_{0}/2>(\mathbb{E}\Vert\mathbf{z} \Vert^{2}/ \inf_{\Vert\mathbf{a} \Vert= 1}
\lambda_{min}(\mathcal{S}(\mathbf{a})))^{1/2}%
\]
and $g(M_{0} / 2) > \sup_{\Vert\mathbf{a} \Vert=1} \text{MSE}_{0}%
(\mathbf{a},\mathbf{0})$. Then
\begin{align*}
\mathbb{P}\left(  \limsup\limits_{T}\Vert\widetilde{\mathbf{D}} \Vert_{F}<
M_{0}\right)  =1.
\end{align*}

\end{lemma}

\begin{proof}
[Proof of Lemma \ref{lemma:B_bounded}]Let
\begin{align*}
C=\left\lbrace \limsup\limits_{T} \sup_{\Vert\mathbf{a} \Vert= 1,
\Vert\mathbf{D}\Vert_{F} < M_{0}} \vert\mathbb{P}_{T}L_{\mathbf{a},\mathbf{D}}
- \mathbb{P}L_{\mathbf{a},\mathbf{D}} \vert= 0\right\rbrace ,\\
D=\left\lbrace \lim\inf\limits_{T}\inf_{\Vert\mathbf{a} \Vert= 1,
\Vert\mathbf{D}\Vert_{F}\geq M_{0}/2} \text{MSE}(\mathbf{a},\mathbf{D}) \geq
g(M_{0} / 2)\right\rbrace ,\\
E=\left\lbrace \limsup\limits_{T}\Vert\widetilde{\mathbf{D}} \Vert_{F}\geq
M_{0}\right\rbrace .
\end{align*}
Assume $\mathbb{P}(E)>0$. Then, by Lemmas \ref{lemma:liminf_MSE} and
\ref{lemma:conv_unif}, $\mathbb{P}(C\cap D \cap E)>0$. Assume in what follows
that we are working in the set $C\cap D \cap E$. Then for sufficiently large
$T$
\begin{align*}
\mathbb{P}_{T}L_{\widetilde{\mathbf{a}},\mathbf{0}}\geq\mathbb{P}%
_{T}L_{\widetilde{\mathbf{a}},\widetilde{\mathbf{D}}}\geq\inf_{\Vert\mathbf{a}
\Vert= 1, \Vert\mathbf{D}\Vert_{F} \geq M_{0}/2}\mathbb{P}_{T}L_{\mathbf{a}%
,\mathbf{D}}.
\end{align*}
It follows that
\begin{align*}
\limsup\limits_{T}\mathbb{P}_{T}L_{\widetilde{\mathbf{a}},\mathbf{0}}%
\geq\liminf\limits_{T}\inf_{\Vert\mathbf{a} \Vert= 1, \Vert\mathbf{D}\Vert_{F}
\geq M_{0}/2}\mathbb{P}_{T}L_{\mathbf{a},\mathbf{D}} \geq g(M_{0} /2).
\end{align*}
It follows easily from $\limsup\limits_{T} \sup_{\Vert\mathbf{a} \Vert= 1,
\Vert\mathbf{D}\Vert_{F} < M_{0}} \vert\mathbb{P}_{T}L_{\mathbf{a},\mathbf{D}}
- \mathbb{P}L_{\mathbf{a},\mathbf{D}} \vert= 0$ that
\begin{align*}
\sup_{\Vert\mathbf{a} \Vert= 1} \text{MSE}_{0}(\mathbf{a},\mathbf{0}) =
\sup_{\Vert\mathbf{a} \Vert= 1} \mathbb{P}L_{\mathbf{a},\mathbf{0}}
\geq\limsup\limits_{T}\mathbb{P}_{T}L_{\widetilde{\mathbf{a}},\mathbf{0}}.
\end{align*}
But by assumption
\begin{align*}
\sup_{\Vert\mathbf{a} \Vert= 1} \text{MSE}_{0}(\mathbf{a},\mathbf{0}) <
g(M_{0} /2).
\end{align*}
We have arrived at a contradiction. It must be that $\mathbb{P}(E)=0$.
\end{proof}

\begin{proof}
[Proof of Proposition \ref{propo:exist_minsamp}]Take $M>0$ such that
\[
g(M)> 2 \inf_{\Vert\mathbf{a} \Vert= 1} \text{MSE}(\mathbf{a},\mathbf{0})
\]
and $M>(\mathbb{E}\Vert\mathbf{z} \Vert^{2}/ \inf_{\Vert\mathbf{a} \Vert= 1}
\lambda_{min}\left(  \mathcal{S}(\mathbf{a}) \right)  )^{1/2}$. Since
$\text{MSE}(\mathbf{a},\mathbf{D})$ is continuous, it attains its minimum over
the compact set $\left\lbrace (\mathbf{a},\mathbf{D}) : \Vert\mathbf{a}
\Vert=1, \Vert\mathbf{D}\Vert_{F} \leq M\right\rbrace $. Moreover
\[
\min_{\Vert\mathbf{a} \Vert= 1, \Vert\mathbf{D}\Vert\leq M} \text{MSE}%
(\mathbf{a},\mathbf{D}) \leq\inf_{\Vert\mathbf{a} \Vert=1} \text{MSE}%
(\mathbf{a},\mathbf{0}).
\]
Let
\[
A = \left\lbrace \liminf\inf_{\Vert\mathbf{a} \Vert= 1, \Vert\mathbf{D}%
\Vert_{F}\geq M} \text{MSE}(\mathbf{a},\mathbf{D}) > g(M) \right\rbrace .
\]
By Lemma \ref{lemma:liminf_MSE}, $\mathbb{P}(A) = 1$. Assume we are working in
the event $A$ in what follows. Then, there exists $T_{0}$ such that for
$T>T_{0}$, $\text{MSE}(\mathbf{a},\mathbf{D}) > \inf_{\Vert\mathbf{a} \Vert=1}
\text{MSE}(\mathbf{a},\mathbf{0})$ for all $\mathbf{a}$, $\mathbf{D} $ with
$\Vert\mathbf{a} \Vert=1$, $\Vert\mathbf{D}\Vert_{F}\geq M$. Hence, for
$T>T_{0}$, $\min_{\Vert\mathbf{a} \Vert= 1, \Vert\mathbf{D}\Vert\leq M}
\text{MSE}(\mathbf{a},\mathbf{D}) \leq\text{MSE}(\widetilde{\mathbf{a}},
\widetilde{\mathbf{D}})$ for all $(\widetilde{\mathbf{a}},
\widetilde{\mathbf{D}})$ with $\Vert\widetilde{\mathbf{a}} \Vert=1$, from
which the results follows.
\end{proof}

\begin{proof}
[Proof of Proposition \ref{propo:exist_minpop}]%
\label{proof:propo:exist_minpop} Fix $\mathbf{a}$ with $\Vert\mathbf{a}
\Vert=1$ and $\mathbf{D}$. Note that
\begin{align*}
\text{MSE}_{0}(\mathbf{a},\mathbf{D}) = \mathbb{E}\Vert\mathbf{z}_{t} -
\widehat{\mathbf{z}}_{t}\Vert^{2} = \mathbb{E}\Vert\widehat{\mathbf{z}}%
_{t}\Vert^{2} + \mathbb{E}\Vert\mathbf{z}_{t}\Vert^{2} - 2\mathbb{E}%
\mathbf{z}_{t}^{\prime} \widehat{\mathbf{z}}_{t}.
\end{align*}
Let $\mathbf{f}_{t}^{\prime}= \mathbf{f}_{t}^{\prime}(\mathbf{a}) =
(1,\mathbf{a}^{\prime} \mathbf{x}_{t}, \mathbf{a}^{\prime} \mathbf{x}%
_{t-1},\dots, \mathbf{a}^{\prime} \mathbf{x}_{t-k_{2}})$. Note that
$\widehat{\mathbf{z}}_{t} = \mathbf{D}^{\prime}\mathbf{f}_{t}$. Hence
\begin{align*}
\mathbb{E}\Vert\widehat{\mathbf{z}}_{t}\Vert^{2} = \mathbb{E}\mathbf{f}%
_{t}^{\prime}\mathbf{D}\mathbf{D}^{\prime}\mathbf{f}_{t}= \Tr(\mathbf{D}%
\mathbf{D}^{\prime} \mathcal{S}(\mathbf{a}) ).
\end{align*}
Since $\mathbf{D}\mathbf{D}^{\prime}$ and $\mathcal{S}(\mathbf{a}) -
\mathbf{I}_{k_{2}+2} \inf_{\Vert\mathbf{a} \Vert= 1} \lambda_{min}%
(\mathcal{S}(\mathbf{a}))$ are symmetric and semi-positive definite we have
that
\begin{align*}
\Tr(\mathbf{D}\mathbf{D}^{\prime} \mathcal{S}(\mathbf{a}) )\geq\Tr(\mathbf{D}%
\mathbf{D}^{\prime} \inf_{\Vert\mathbf{a} \Vert= 1} \lambda_{min}%
(\mathcal{S}(\mathbf{a}))) = \Vert\mathbf{D} \Vert_{F}^{2}\inf_{\Vert
\mathbf{a} \Vert= 1} \lambda_{min}(\mathcal{S}(\mathbf{a})).
\end{align*}
On the other hand, by the Cauchy-Schwartz inequality
\begin{align*}
\mathbb{E}\vert\mathbf{z}_{t}^{\prime}\widehat{\mathbf{z}}_{t}\vert
\leq(\mathbb{E}\Vert\mathbf{z}_{t}\Vert^{2})^{1/2}(\mathbb{E}\Vert
\widehat{\mathbf{z}}_{t}\Vert^{2})^{1/2}.
\end{align*}
Since $\mathbf{D}\mathbf{D}^{\prime}$ and $\mathbf{I}_{k_{2}+2} \sup
_{\Vert\mathbf{a} \Vert= 1} \lambda_{max}(\mathcal{S}(\mathbf{a})) -
\mathcal{S}(\mathbf{a})$ are symmetric and semi-positive definite we have
that
\begin{align*}
\Tr(\mathbf{D}\mathbf{D}^{\prime} \mathcal{S}(\mathbf{a}) )\leq\Tr(\mathbf{D}%
\mathbf{D}^{\prime} \sup_{\Vert\mathbf{a} \Vert= 1} \lambda_{max}%
(\mathcal{S}(\mathbf{a}))) = \Vert\mathbf{D} \Vert_{F}^{2}\sup_{\Vert
\mathbf{a} \Vert= 1} \lambda_{max}(\mathcal{S}(\mathbf{a})).
\end{align*}
Hence,
\begin{align*}
(\mathbb{E}\Vert\widehat{\mathbf{z}}_{t}\Vert^{2})^{1/2}\leq\Vert\mathbf{D}
\Vert_{F} (\sup_{\Vert\mathbf{a} \Vert= 1} \lambda_{max}(\mathcal{S}%
(\mathbf{a})))^{1/2}.
\end{align*}
It follows that
\begin{align*}
\inf_{\Vert\mathbf{a} \Vert= 1} \mathbb{E}\Vert\mathbf{z}_{t} -
\widehat{\mathbf{z}}_{t}\Vert^{2} \geq\Vert\mathbf{D} \Vert_{F}^{2}\inf
_{\Vert\mathbf{a} \Vert= 1} \lambda_{min}(\mathcal{S}(\mathbf{a})) +
\mathbb{E}\Vert\mathbf{z}_{t}\Vert^{2} - 2 (\mathbb{E}\Vert\mathbf{z}_{t}%
\Vert^{2})^{1/2} \Vert\mathbf{D} \Vert_{F} (\sup_{\Vert\mathbf{a} \Vert= 1}
\lambda_{max}(\mathcal{S}(\mathbf{a})))^{1/2},
\end{align*}
from which the proposition follows immediately.
\end{proof}

\begin{proof}
[Proof of Theorem \ref{theo:consistency}]Fix $\varepsilon>0$. Let $A=\left\{
\limsup\limits_{T}d((\widehat{\mathbf{a}},\widehat{\mathbf{D}}),\mathcal{I}%
)\geq\varepsilon\right\}  $. We will show that $\mathbb{P}(A)=0$. Assume
$\mathbb{P}(A)>0$. Take $(\mathbf{a}^{\ast},\mathbf{D}^{\ast})\in\mathcal{I}$.
Fix $M_{0}$ large enough such that $M_{0}/2$ satisfies the hypothesis of Lemma
\ref{lemma:B_bounded}. Note that
\[
\inf\left\{  \text{MSE}_{0}(\mathbf{a},\mathbf{D}):\Vert\mathbf{a}%
\Vert=1,\Vert\mathbf{D}\Vert_{F}\leq M_{0},d((\mathbf{a},\mathbf{D}%
),\mathcal{I})\geq\varepsilon/2\right\}  >\text{MSE}_{0}(\mathbf{a}^{\ast
},\mathbf{D}^{\ast}).
\]
Since $\mathbb{P}_{T}L_{\widehat{\mathbf{a}},\widehat{\mathbf{D}}}%
\leq\mathbb{P}_{T}L_{\mathbf{a}^{\ast},\mathbf{D}^{\ast}}$, and by the Ergodic
Theorem
\[
\mathbb{P}_{T}L_{\mathbf{a}^{\ast},\mathbf{D}^{\ast}}%
\overset{a.s.}{\rightarrow}\text{MSE}_{0}(\mathbf{a}^{\ast},\mathbf{D}^{\ast
}),
\]
we have that, with probability one
\[
\limsup\limits_{T}\mathbb{P}_{T}L_{\widehat{\mathbf{a}},\widehat{\mathbf{D}}%
}\leq\text{MSE}_{0}(\mathbf{a}^{\ast},\mathbf{D}^{\ast}).
\]
Let
\begin{align*}
B  &  =\left\{  \limsup\limits_{T}\sup_{\Vert\mathbf{a}\Vert=1,\Vert
\mathbf{D}\Vert_{F}\leq M_{0}}|\mathbb{P}_{T}L_{\mathbf{a},\mathbf{D}%
}-\mathbb{P}L_{\mathbf{a},\mathbf{D}}|=0\right\}  ,\\
C  &  =\left\{  \limsup\limits_{T}\Vert\widehat{\mathbf{D}}\Vert_{F}%
<M_{0}/2\right\}  .\\
D  &  =\left\{  \limsup\limits_{T}\mathbb{P}_{T}L_{\widehat{\mathbf{a}%
},\widehat{\mathbf{D}}}\leq\text{MSE}_{0}(\mathbf{a}^{\ast},\mathbf{D}^{\ast
})\right\}  .
\end{align*}
Then, by Lemmas \ref{lemma:B_bounded} and \ref{lemma:conv_unif},
$\mathbb{P}(A\cap B\cap C\cap D)>0$. Assume in what follows that we are
working in the set $A\cap B\cap C\cap D$. Note that since
\[
\left\{  (\mathbf{a},\mathbf{D}):\Vert\mathbf{a}\Vert=1,\Vert\mathbf{D}%
\Vert_{F}\leq M_{0},d((\mathbf{a},\mathbf{D}),\mathcal{I})\geq\varepsilon
/2\right\}  \subseteq\left\{  (\mathbf{a},\mathbf{D}):\Vert\mathbf{a}%
\Vert=1,\Vert\mathbf{D}\Vert_{F}\leq M_{0}\right\}
\]
we have that
\[
\limsup\limits_{T}\sup\left\{  |\mathbb{P}_{T}L_{\mathbf{a},\mathbf{D}%
}-\mathbb{P}L_{\mathbf{a},\mathbf{D}}|:|\Vert\mathbf{a}\Vert=1,\Vert
\mathbf{D}\Vert_{F}\leq M_{0},d((\mathbf{a},\mathbf{D}),\mathcal{I}%
)\geq\varepsilon/2\right\}  =0,
\]
and hence that
\begin{align*}
&  \liminf\limits_{T}\inf\left\{  \mathbb{P}_{T}L_{\mathbf{a},\mathbf{D}%
}:\Vert\mathbf{a}\Vert=1,\Vert\mathbf{D}\Vert_{F}\leq M_{0},d((\mathbf{a}%
,\mathbf{D}),\mathcal{I})\geq\varepsilon/2\right\} \\
&  \geq\inf\left\{  \mathbb{P}L_{\mathbf{a},\mathbf{D}}:\Vert\mathbf{a}%
\Vert=1,\Vert\mathbf{D}\Vert_{F}\leq M_{0},d((\mathbf{a},\mathbf{D}%
),\mathcal{I})\geq\varepsilon/2\right\}  .
\end{align*}
Since $\limsup\limits_{T}d((\widehat{\mathbf{a}},\widehat{\mathbf{D}%
}),\mathcal{I})\geq\varepsilon$ and $\limsup\limits_{T}\Vert
\widehat{\mathbf{D}}\Vert_{F}<M_{0}/2$, there exists a subsequence of
$(\widehat{\mathbf{a}},\widehat{\mathbf{D}})$, which in an abuse of notation
we continue to call $(\widehat{\mathbf{a}},\widehat{\mathbf{D}})$, such that
$d((\widehat{\mathbf{a}},\widehat{\mathbf{D}}),\mathcal{I})>\varepsilon/2$ and
$\Vert\widehat{\mathbf{D}}\Vert_{F}\leq M_{0}$ for all $T$. Then
\begin{align*}
\text{MSE}_{0}(\mathbf{a}^{\ast},\mathbf{D}^{\ast})  &  \geq\liminf
\limits_{T}\mathbb{P}_{T}L_{\widehat{\mathbf{a}},\widehat{\mathbf{B}}}\\
&  \geq\liminf\limits_{T}\inf\left\{  \mathbb{P}_{T}L_{\mathbf{a},\mathbf{D}%
}:\Vert\mathbf{a}\Vert=1,\Vert\mathbf{D}\Vert_{F}\leq M_{0},d((\mathbf{a}%
,\mathbf{D}),\mathcal{I})\geq\varepsilon/2\right\} \\
&  \geq\inf\left\{  \mathbb{P}L_{\mathbf{a},\mathbf{D}}:\Vert\mathbf{a}%
\Vert=1,\Vert\mathbf{D}\Vert_{F}\leq M_{0},d((\mathbf{a},\mathbf{D}%
),\mathcal{I})\geq\varepsilon/2\right\} \\
&  =\inf\left\{  \text{MSE}_{0}(\mathbf{a},\mathbf{D}):\Vert\mathbf{a}%
\Vert=1,\Vert\mathbf{D}\Vert_{F}\leq M_{0},d((\mathbf{a},\mathbf{D}%
),\mathcal{I})\geq\varepsilon/2\right\} \\
&  >\text{MSE}_{0}(\mathbf{a}^{\ast},\mathbf{D}^{\ast}),
\end{align*}
a contradiction. It must be that $\mathbb{P}(A)=0$.
\end{proof}

\begin{proof}
[Proof of Theorem \ref{theo:dfm}]For all $\mathbf{a}\in\mathbb{R}^{m(k_{1}%
+1)}$
\begin{align*}
\widehat{\mathbf{f}}_{t}(\mathbf{a})  &  =%
\begin{pmatrix}
\mathbf{z}_{t}^{\prime} & \dots & \mathbf{z}_{t-k_{1}}^{\prime}\\
\vdots & \vdots & \vdots\\
\mathbf{z}_{t-k_{2}}^{\prime} & \dots & \mathbf{z}_{t-k_{1}-k_{2}}^{\prime}\\
&  &
\end{pmatrix}
\mathbf{a}\\
&  =%
\begin{pmatrix}
\mathbf{f}_{t}^{\prime}\mathbf{B} & \dots & \mathbf{f}_{t-k_{1}}^{\prime
}\mathbf{B}\\
\vdots & \vdots & \vdots\\
\mathbf{f}_{t-k_{2}}^{\prime}\mathbf{B} & \dots & \mathbf{f}_{t-k_{1}-k_{2}%
}^{\prime}\mathbf{B}\\
&  &
\end{pmatrix}
\mathbf{a}+%
\begin{pmatrix}
\mathbf{e}_{t}^{\prime} & \dots & \mathbf{e}_{t-k_{1}}^{\prime}\\
\vdots & \vdots & \vdots\\
\mathbf{e}_{t-k_{2}}^{\prime} & \dots & \mathbf{e}_{t-k_{1}-k_{2}}^{\prime}\\
&  &
\end{pmatrix}
\mathbf{a}\\
&  =%
\begin{pmatrix}
\mathbf{f}_{t}^{\prime} & \dots & \mathbf{f}_{t-k_{1}}^{\prime}\\
\vdots & \vdots & \vdots\\
\mathbf{f}_{t-k_{2}}^{\prime} & \dots & \mathbf{f}_{t-k_{1}-k_{2}}^{\prime}\\
&  &
\end{pmatrix}
(\mathbf{I}_{k_{1}+1}\otimes\mathbf{B})\mathbf{a}+%
\begin{pmatrix}
\mathbf{e}_{t}^{\prime} & \dots & \mathbf{e}_{t-k_{1}}^{\prime}\\
\vdots & \vdots & \vdots\\
\mathbf{e}_{t-k_{2}}^{\prime} & \dots & \mathbf{e}_{t-k_{1}-k_{2}}^{\prime}\\
&  &
\end{pmatrix}
\mathbf{a}\\
&  =\mathbf{F}_{t}(\mathbf{I}_{k_{1}+1}\otimes\mathbf{B})\mathbf{a}%
+\mathbf{E}_{t}\mathbf{a}.
\end{align*}
Let $\widetilde{\mathbf{a}}\in\mathbb{R}^{m(k_{1}+1)}$ be defined by
$\widetilde{\mathbf{a}}=((1/\Vert\mathbf{b}_{0}\Vert)\mathbf{b}_{0}^{\prime
},\mathbf{0}_{m},\dots,\mathbf{0}_{m})^{\prime}$. Since $(\mathbf{a}^{\ast
},\mathbf{B}^{\ast})\in\mathcal{I}$
\begin{align*}
\text{MSE}_{0}(\mathbf{a}^{\ast},\mathbf{B}^{\ast})  &  \leq\text{MSE}%
_{0}(\widetilde{\mathbf{a}},\mathbf{B}/\Vert\mathbf{b}_{0}\Vert)\\
&  =\mathbb{E}\left\Vert \mathbf{B}^{\prime}\mathbf{f}_{t}+\mathbf{e}%
_{t}-(\mathbf{B}^{\prime}/\Vert\mathbf{b}_{0}\Vert)\mathbf{F}_{t}%
(\mathbf{I}_{k_{1}+1}\otimes\mathbf{B})\widetilde{\mathbf{a}}-(\mathbf{B}%
^{\prime}/\Vert\mathbf{b}_{0}\Vert)\mathbf{E}_{t}\widetilde{\mathbf{a}%
}\right\Vert ^{2}\\
&  =\mathbb{E}\left\Vert \mathbf{B}^{\prime}\mathbf{f}_{t}-\mathbf{B}^{\prime
}\mathbf{F}_{t}(\mathbf{I}_{k_{1}+1}\otimes(\mathbf{B}/\Vert\mathbf{b}%
_{0}\Vert))\widetilde{\mathbf{a}}\right\Vert ^{2}+\mathbb{E}\left\Vert
\mathbf{e}_{t}-(\mathbf{B}^{\prime}/\Vert\mathbf{b}_{0}\Vert)\mathbf{E}%
_{t}\widetilde{\mathbf{a}}\right\Vert ^{2}\\
&  +2\mathbb{E}\left(  \mathbf{B}^{\prime}\mathbf{f}_{t}-\mathbf{B}^{\prime
}\mathbf{F}_{t}(\mathbf{I}_{k_{1}+1}\otimes(\mathbf{B}/\Vert\mathbf{b}%
_{0}\Vert))\widetilde{\mathbf{a}}\right)  ^{\prime}\left(  \mathbf{e}%
_{t}-(\mathbf{B}^{\prime}/\Vert\mathbf{b}_{0}\Vert)\mathbf{E}_{t}%
\widetilde{\mathbf{a}}\right)  .
\end{align*}
By Condition \ref{condition:DFM}b) and c), $\mathbb{E}\left(  \mathbf{B}%
^{\prime}\mathbf{f}_{t}-\mathbf{B}^{\prime}\mathbf{F}_{t}(\mathbf{I}_{k_{1}%
+1}\otimes(\mathbf{B}/\Vert\mathbf{b}_{0}\Vert))\widetilde{\mathbf{a}}\right)
^{\prime}\left(  \mathbf{e}_{t}-(\mathbf{B}^{\prime}/\Vert\mathbf{b}_{0}%
\Vert)\mathbf{E}_{t}\widetilde{\mathbf{a}}\right)  =0$. By Condition
\ref{condition:DFM}a)
\begin{align*}
\mathbb{E}\left\Vert \mathbf{B}^{\prime}\mathbf{f}_{t}-\mathbf{B}^{\prime
}\mathbf{F}_{t}(\mathbf{I}_{k_{1}+1}\otimes(\mathbf{B}/\Vert\mathbf{b}%
_{0}\Vert))\widetilde{\mathbf{a}}\right\Vert ^{2}  &  \leq\Vert\mathbf{B}%
\Vert^{2}\mathbb{E}\left\Vert \mathbf{f}_{t}-\mathbf{F}_{t}(\mathbf{I}%
_{k_{1}+1}\otimes(\mathbf{B}/\Vert\mathbf{b}_{0}\Vert))\widetilde{\mathbf{a}%
}\right\Vert ^{2}\\
&  =O(m)\mathbb{E}\left\Vert \mathbf{f}_{t}-\mathbf{F}_{t}(\mathbf{I}%
_{k_{1}+1}\otimes(\mathbf{B}/\Vert\mathbf{b}_{0}\Vert))\widetilde{\mathbf{a}%
}\right\Vert ^{2}.
\end{align*}
By Condition \ref{condition:DFM}a), $\Vert\mathbf{b}_{0}\Vert^{2}%
/m\rightarrow1$. Then
\begin{align*}
(\mathbf{I}_{k_{1}+1}\otimes(\mathbf{B}/\Vert\mathbf{b}_{0}\Vert
))\widetilde{\mathbf{a}}  &  =(1,(\mathbf{b}_{1}^{\prime}\mathbf{b}_{0}%
)/\Vert\mathbf{b}_{0}\Vert^{2},(\mathbf{b}_{2}^{\prime}\mathbf{b}_{0}%
)/\Vert\mathbf{b}_{0}\Vert^{2},\dots,0)\\
&  \rightarrow(1,0,\dots,0)^{\prime}\in\mathbb{R}^{(k_{1}+1)(k_{2}+1)}%
\end{align*}
which implies that $\mathbf{F}_{t}(\mathbf{I}_{k_{1}+1}\otimes(\mathbf{B}%
/\Vert\mathbf{b}_{0}\Vert))\widetilde{\mathbf{a}}\rightarrow\mathbf{f}_{t} $.

It follows from Condition \ref{condition:DFM}b) and the Bounded Convergence
Theorem that
\begin{align*}
\frac{1}{m}\mathbb{E}\left\Vert \mathbf{B}^{\prime} \mathbf{f}_{t} -
\mathbf{B}^{\prime}\mathbf{F}_{t} (\mathbf{I}_{k_{1} + 1} \otimes
(\mathbf{B}/\Vert\mathbf{b}_{0}\Vert))\widetilde{\mathbf{a}} \right\Vert ^{2}
\rightarrow0.
\end{align*}
Note that
\begin{align*}
\mathbb{E} \left\Vert \mathbf{e}_{t} - (\mathbf{B}^{\prime}/\Vert
\mathbf{b}_{0}\Vert)\mathbf{E}_{t}\widetilde{\mathbf{a}} \right\Vert ^{2} =
\mathbb{E}\left\Vert \mathbf{e}_{t}\right\Vert ^{2} + \mathbb{E}\left\Vert
(\mathbf{B}^{\prime}/\Vert\mathbf{b}_{0}\Vert)\mathbf{E}_{t}%
\widetilde{\mathbf{a}} \right\Vert ^{2} - 2 \mathbb{E} \mathbf{e}_{t}^{\prime
}(\mathbf{B}^{\prime}/\Vert\mathbf{b}_{0}\Vert)\mathbf{E}_{t}%
\widetilde{\mathbf{a}}.
\end{align*}
Now $\mathbf{E}_{t}\widetilde{\mathbf{a}} = ((\mathbf{e}_{t}^{\prime}
\mathbf{b}_{0})/\Vert\mathbf{b}_{0}\Vert,\dots, (\mathbf{e}_{t-k_{2}}^{\prime}
\mathbf{b}_{0})/\Vert\mathbf{b}_{0}\Vert)^{\prime}$ and hence
\begin{align*}
\Vert\mathbf{E}_{t}\widetilde{\mathbf{a}} \Vert^{2}  &  = \sum\limits_{h=0}%
^{k_{2}} (\mathbf{e}_{t-h}^{\prime} \mathbf{b}_{0})^{2}/\Vert\mathbf{b}%
_{0}\Vert^{2}.
\end{align*}
It follows that $\mathbb{E}\Vert\mathbf{E}_{t}\widetilde{\mathbf{a}} \Vert^{2}
= (k_{2} + 1) \left(  (1/\Vert\mathbf{b}_{0}\Vert^{2}) \mathbf{b}_{0}^{\prime}
\boldsymbol{\Sigma}^{e}(0)\mathbf{b}_{0}\right)  $. Hence, using Condition
\ref{condition:DFM}a) and c)
\begin{align*}
&  \frac{1}{m}\mathbb{E}\Vert(\mathbf{B}^{\prime}/\Vert\mathbf{b}_{0}\Vert)
\mathbf{E}_{t}\widetilde{\mathbf{a}} \Vert^{2} \leq\frac{\Vert\mathbf{B}%
^{\prime}\Vert^{2}}{m}(k_{2} + 1) \left(  (1/\Vert\mathbf{b}_{0}\Vert^{4})
\mathbf{b}_{0}^{\prime} \boldsymbol{\Sigma}^{e}(0)\mathbf{b}_{0}\right)  =
o(1).
\end{align*}
Finally,
\begin{align*}
\frac{1}{m\Vert\mathbf{b}_{0}\Vert}\mathbb{E} \vert\mathbf{e}_{t}^{\prime
}\mathbf{B}^{\prime}\mathbf{E}_{t}\widetilde{\mathbf{a}} \vert\leq\left(
\mathbb{E} \Vert\mathbf{e}_{t}\Vert^{2}/m\right)  ^{1/2} \left(
\mathbb{E}\Vert(\mathbf{B}^{\prime}/\Vert\mathbf{b}_{0}\Vert)\mathbf{E}%
_{t}\widetilde{\mathbf{a}} \Vert^{2}/m\right)  ^{1/2}=o(1).
\end{align*}
Since $\mathbb{E}\Vert\mathbf{e}_{t}\Vert^{2}=\sum_{j=1}^{m}\mathbb{E}
e_{t,j}^{2}$, the Theorem is proven.
\end{proof}

We will need the following general results on linear predictors.

\begin{lemma}
\label{lemma:bound_cov} Let $X$ and $Z_{1}, \dots, Z_{n}$ be zero mean random
variables satisfying $\mathbb{E}Z_{i}^{2}\leq b$, for some $b>0$ and
$i=1,\dots,n$, and $\mathbb{E}Z_{i}Z_{j}=0$ for $i\neq j$. Let $P(X \vert
Z_{i})$ be the best linear predictor of $X$ based on $Z_{i}$. Then
\[
\sum\limits_{i=1}^{n} \mathbb{E}P(Z_{i}\vert X)^{2} \leq b.
\]

\end{lemma}

\begin{proof}
[Proof of Lemma \ref{lemma:bound_cov}]Let $P(X \vert Z)$ be the best linear
predictor of $X$ based on $Z_{1}, \dots Z_{n}$. Since $Z_{1}, \dots, Z_{n}$
are uncorrelated, we have
\begin{align*}
P(X \vert Z) = \sum\limits_{i=1}^{n} P(X \vert Z_{i})= \sum\limits_{i=1}^{n}
\frac{\mathbb{E}X Z_{i}}{\mathbb{E}Z_{i}^{2}} Z_{i}.
\end{align*}
Hence
\begin{align*}
\mathbb{E}X^{2} \geq\mathbb{E}P(X \vert Z)^{2} = \sum\limits_{i=1}^{n}
\mathbb{E}P(X \vert Z_{i})^{2}=\sum\limits_{i=1}^{n} \frac{(\mathbb{E}X
Z_{i})^{2}}{\mathbb{E}Z_{i}^{2}}\geq(1/b)\sum\limits_{i=1}^{n} (\mathbb{E}X
Z_{i})^{2}%
\end{align*}
Now,
\begin{align*}
\sum\limits_{i=1}^{n} \mathbb{E}P(Z_{i}\vert X)^{2} =\sum\limits_{i=1}^{n}
\mathbb{E}(\frac{\mathbb{E}X Z_{i}}{\mathbb{E}X^{2}}X)^{2} = \sum
\limits_{i=1}^{n} \frac{(\mathbb{E}X Z_{i})^{2}}{\mathbb{E}X^{2}} \leq b.
\end{align*}

\end{proof}

\begin{lemma}
\label{lemma:pred_ideo} Let $Z_{1}, \dots, Z_{n}$ be zero mean random
variables satisfying $\mathbb{E}Z_{i}^{2}\leq b$, for some $b>0$ and
$i=1,\dots,n$, and $\mathbb{E}Z_{i}Z_{j}=0$ for $i\neq j$. Let $Y_{1}, \dots,
Y_{k}$ be zero mean random variables and $P(Z_{i} \vert Y)$ be the best linear
predictor of $Z_{i}$ based on $Y_{1}, \dots, Y_{k}$. Then
\begin{align*}
\liminf_{n\rightarrow\infty} \frac{1}{n}\sum\limits_{i=1}^{n} \left(
\mathbb{E}\left(  Z_{i} - P(Z_{i} \vert Y) \right)  ^{2} -\mathbb{E}Z_{i}^{2}
\right)  \geq0.
\end{align*}

\end{lemma}

\begin{proof}
[Proof of Lemma \ref{lemma:pred_ideo}]Let $X_{1}, \dots, X_{k}$ be zero mean,
unit variance, uncorrelated random variables with the same linear span as
$Y_{1}, \dots, Y_{k}$. Let $P(Z_{i} \vert X)$ be the best linear predictor of
$Z_{i}$ based on $X_{1}, \dots, X_{k}$. Then, using Lemma
\ref{lemma:bound_cov}
\begin{align*}
\sum\limits_{i=1}^{n}\mathbb{E}\left(  Z_{i} - P(Z_{i} \vert Y) \right)  ^{2}
&  = \sum\limits_{i=1}^{n}\mathbb{E}\left(  Z_{i} - P(Z_{i} \vert X) \right)
^{2} = \sum\limits_{i=1}^{n} \left(  \mathbb{E}Z_{i}^{2} - \mathbb{E}P(Z_{i}
\vert X)^{2}\right) \\
&  = \sum\limits_{i=1}^{n} \mathbb{E}Z_{i}^{2} - \sum\limits_{i=1}^{n}%
\sum\limits_{j=1}^{k}\mathbb{E}P(Z_{i} \vert X_{j})^{2} \geq\sum
\limits_{i=1}^{n} \mathbb{E}Z_{i}^{2} - k b.
\end{align*}
Hence
\begin{align*}
\liminf_{n\rightarrow\infty} \frac{1}{n}\sum\limits_{i=1}^{n} \left(
\mathbb{E}\left(  Z_{i} - P(Z_{i} \vert Y) \right)  ^{2} - \mathbb{E}Z_{i}^{2}
\right)  \geq\lim_{n\rightarrow\infty} -\frac{kb}{n} = 0.
\end{align*}

\end{proof}

\begin{proof}
[Proof of Theorem \ref{theo:common}]%
\begin{align*}
\frac{1}{m}\text{MSE}_{0}(\mathbf{a}^{\ast},\mathbf{B}^{\ast})  &  =\frac
{1}{m}\mathbb{E}\Vert\mathbf{z}_{t}-\widehat{\mathbf{z}}_{t}\Vert^{2}=\frac
{1}{m}\mathbb{E}\Vert\mathbf{B}^{\prime}\mathbf{f}_{t}-\widehat{\mathbf{z}%
}_{t}+\mathbf{e}_{t}\Vert^{2}\\
&  =\frac{1}{m}\mathbb{E}\Vert\mathbf{B}^{\prime}\mathbf{f}_{t}%
-\widehat{\mathbf{z}}_{t}\Vert^{2}+\frac{1}{m}\mathbb{E}\Vert\mathbf{e}%
_{t}\Vert^{2}+2\mathbb{E}\mathbf{e}_{t}^{\prime}(\mathbf{B}^{\prime}%
\mathbf{f}_{t}-\widehat{\mathbf{z}}_{t}).
\end{align*}
By Condition \ref{condition:DFM}c)
\[
\mathbb{E}\mathbf{e}_{t}^{\prime}(\mathbf{B}^{\prime}\mathbf{f}_{t}%
-\widehat{\mathbf{z}}_{t})=\mathbb{E}\mathbf{e}_{t}^{\prime}%
\widehat{\mathbf{z}}_{t}.
\]
By Theorem \ref{theo:dfm}, it suffices to show that
\[
\frac{1}{m}\mathbb{E}\mathbf{e}_{t}^{\prime}\widehat{\mathbf{z}}%
_{t}\rightarrow0.
\]
Let $\widehat{e}_{t,j}$ be the best linear predictor of $e_{t,j}$ based on
$\widehat{z}_{t,j}$. Then
\[
\mathbb{E}\left(  e_{t,j}-\widehat{e}_{t,j}\right)  ^{2}=\mathbb{E}e_{t,j}%
^{2}-\frac{(\mathbb{E}\widehat{z}_{t,j}e_{t,j})^{2}}{\mathbb{E}\widehat{z}%
_{t,j}^{2}}.
\]
Note that $\mathbb{E}\widehat{z}_{t,j}^{2}\leq\mathbb{E}z_{t,j}^{2}\leq L$,
since $\widehat{z}_{t,j}$ is obtained by projecting $z_{t,j}$ on the space
spanned by $\widehat{\mathbf{f}}_{t}$. Then
\[
\frac{1}{m}\sum\limits_{j=1}^{m}\mathbb{E}\left(  e_{t,j}-\widehat{e}%
_{t,j}\right)  ^{2}\leq\frac{1}{m}\sum\limits_{j=1}^{m}\left(  \mathbb{E}%
e_{t,j}^{2}-\frac{(\mathbb{E}\widehat{z}_{t,j}e_{t,j})^{2}}{L}\right)  .
\]
Since by Lemma \ref{lemma:pred_ideo}, $\liminf(1/m)\sum_{j=1}^{m}\left(
\mathbb{E}\left(  e_{t,j}-\widehat{e}_{t,j}\right)  ^{2}-\mathbb{E}e_{t,j}%
^{2}\right)  \geq0$ it must be that $$(1/m)\sum_{j=1}^{m}(\mathbb{E}%
\widehat{z}_{t,j}e_{t,j})^{2}\rightarrow0,$$ from which the result follows by
applying the Cauchy-Schwartz inequality.
\end{proof}

\bibliographystyle{apalike}
\bibliography{fore}

\end{document}